\documentclass[a4paper,USenglish,cleveref, autoref, thm-restate]{lipics-v2021}

\usepackage{bbm}
\usepackage{cancel}
\usepackage{csquotes}
\usepackage{mathtools}
\usepackage{multicol}
\usepackage{todonotes}

\usepackage{tikz}
\usetikzlibrary{arrows.meta,
                calc,
                decorations,
                decorations.pathreplacing,
                positioning,
                shapes}

\tikzset{vertex/.style={circle, fill=black, inner sep=2pt}
}

\title{Graph Similarity and Homomorphism Densities}

\author{Jan Böker}{RWTH Aachen University, Aachen, Germany}{boeker@informatik.rwth-aachen.de}{https://orcid.org/0000-0003-4584-121X}{}

\authorrunning{J.\ Böker} 

\Copyright{Jan Böker} 



\relatedversion{}

\acknowledgements{The author would like to thank Yufei Zhao for pointing out how a non-quantitative inverse counting lemma for the cut distance follows from the compactness of the graphon space and the anonymous reviewers for their helpful comments.}

\nolinenumbers 

\hideLIPIcs  

\EventEditors{John Q. Open and Joan R. Access}
\EventNoEds{2}
\EventLongTitle{42nd Conference on Very Important Topics (CVIT 2016)}
\EventShortTitle{CVIT 2016}
\EventAcronym{CVIT}
\EventYear{2016}
\EventDate{December 24--27, 2016}
\EventLocation{Little Whinging, United Kingdom}
\EventLogo{}
\SeriesVolume{42}
\ArticleNo{23}

\newcommand{\dist}{\delta_\square}
\newcommand{\fracdist}{d_\square}
\newcommand{\dens}{t}
\newcommand{\Hom}{\mathsf{Hom}}
\renewcommand{\hom}{\mathsf{hom}}

\newcommand{\N}{\mathbb{N}}
\newcommand{\R}{\mathbb{R}}

\newcommand{\Rnn}{\mathbb{R}_{\ge 0}}
\newcommand{\es}{\mathsf{e}}
\newcommand{\vs}{\mathsf{v}}
\newcommand{\numTo}[1]{[#1]}

\renewcommand{\epsilon}{\varepsilon}

\DeclareMathOperator{\lcm}{lcm}

\newcommand{\trees}{\mathcal{T}}

\newcommand{\groheDist}{\text{dist}_{\lVert \cdot \rVert}}\newcommand{\Fiso}{\mathsf{F}_{\text{iso}}(G, H)}

\newcommand{\neiDistance}{\delta^\mathcal{T}_{\square}}\newcommand{\neiDist}{\neiDistance}\newcommand{\neiDistNo}{\delta^\mathcal{T}}

\newcommand{\kernels}{\mathcal{W}}\newcommand{\graphons}{\mathcal{W}_0}\newcommand{\graphonsWI}{\widetilde{\mathcal{W}}_0}\newcommand{\measPres}{S_{[0,1]}}\newcommand{\intd}[1]{\,d#1}\newcommand{\dx}{\intd{x}}\newcommand{\dy}{\intd{y}}\newcommand{\dz}{\intd{z}}\newcommand{\markov}{\mathcal{M}}\newcommand{\normC}[1]{\lVert #1 \rVert_\square}\newcommand{\normCT}[1]{\lVert #1 \rVert_{\square,2}}\newcommand{\normO}[1]{\lVert #1 \rVert_1}\newcommand{\normI}[1]{\lVert #1 \rVert_\infty}\newcommand{\normT}[1]{\lVert #1 \rVert_2}\newcommand{\normTT}[1]{\lVert #1 \rVert_{\C, 2\rightarrow2}}\newcommand{\normRTT}[1]{\lVert #1 \rVert_{2\rightarrow2}}\newcommand{\normpq}[1]{\lVert #1 \rVert_{\C, p\rightarrow q}}\newcommand{\normRpq}[1]{\lVert #1 \rVert_{p\rightarrow q}}\newcommand{\normIO}[1]{\lVert #1 \rVert_{\C, \infty\rightarrow1}}\newcommand{\normRIO}[1]{\lVert #1 \rVert_{\infty\rightarrow1}}\newcommand{\normOI}[1]{\lVert #1 \rVert_{\C, 1\rightarrow\infty}}\newcommand{\normOO}[1]{\lVert #1 \rVert_{\C, 1\rightarrow 1}}\newcommand{\normII}[1]{\lVert #1 \rVert_{\C, \infty\rightarrow\infty}}\newcommand\restr[2]{{\left.\kern-\nulldelimiterspace #1 \vphantom{\big|} \right|_{#2} }}

\newcommand{\treeDist}{\delta^\mathcal{T}}\newcommand{\someClass}{\mathcal{F}}\newcommand{\someDist}{\delta^\someClass}

\newcommand{\C}{\mathbb{C}}\newcommand{\Lp}{L_p[0,1]}\newcommand{\Lq}{L_q[0,1]}\newcommand{\LO}{L_1[0,1]}\newcommand{\LT}{L_2[0,1]}\newcommand{\LI}{L_\infty[0,1]}\newcommand{\allOne}{\boldsymbol{1}}
\newcommand{\cFun}{\mathbbm{1}}
\newcommand{\supFG}{\sup_{f, g \colon [0,1] \to [0,1]}}\newcommand{\neiDistEu}{\delta^\mathcal{T}_{\square, 2}}\newcommand{\neiDistOp}{\delta^\mathcal{T}_{\C, 2 \rightarrow 2}}\newcommand{\neiDistIO}{\delta^\mathcal{T}_{\C, \infty \rightarrow 1}}\newcommand{\neiDistIOR}{\delta^\mathcal{T}_{\infty \rightarrow 1}}\newcommand{\neiDistpq}{\delta^\mathcal{T}_{\C, p \rightarrow q}}\newcommand{\neiDistRpq}{\delta^\mathcal{T}_{p \rightarrow q}}\renewcommand{\Re}{\operatorname{Re}}\renewcommand{\Im}{\operatorname{Im}}\newcommand{\iu}{{i\mkern1mu}}\newcommand{\neiDistSpec}{\delta^{\mathcal{T}}_{2}}\newcommand{\neiDistROp}{\delta^\mathcal{T}_{2 \rightarrow 2}}

\newcommand{\pathDistNoNorm}{\delta^\mathcal{P}}\newcommand{\pathDist}{\delta^\mathcal{P}_{\square}}\newcommand{\signedMarkovs}{\mathcal{S}}\newcommand{\pathDistROp}{\delta_{2 \rightarrow 2}^\mathcal{P}}\newcommand{\pathDistSpec}{\delta_{2}^\mathcal{P}}

\newcommand{\ColG}{G/C_\infty^G}\newcommand{\Col}[1]{{#1}/C_\infty^{#1}}\newcommand{\colDist}{\dist^{\mathcal{C}}}

\begin{document}

\maketitle

\begin{abstract}
    We introduce the tree distance,
    a new distance measure on graphs.
    The tree distance can be
    computed in polynomial time with standard methods
    from convex optimization.
    It is based on the notion of fractional isomorphism,
    a characterization based on a natural system of linear equations
    whose integer solutions correspond to graph isomorphism.
    By results of Tinhofer ($1986$, $1991$) and Dvořák ($2010$),
    two graphs $G$ and $H$ are fractionally isomorphic
    if and only if, for every tree $T$,
    the number of homomorphisms from $T$ to $G$
    equals the corresponding number from $T$ to $H$,
    which means that the tree distance of $G$ and $H$ is zero.
    Our main result is that this
    correspondence between the equivalence relations
    \enquote{fractional isomorphism}
    and \enquote{equal tree homomorphism densities}
    can be extended to a correspondence
    between the associated distance measures.
    Our result is inspired by a similar result due to
    Lovász and Szegedy ($2006$) and
Borgs, Chayes, Lovász, Sós, and Vesztergombi ($2008$)
    that connects the cut distance of graphs to their homomorphism
    densities (over all graphs), which is a
    fundamental theorem in the theory of graph limits.
    We also introduce the path distance of graphs and take
    the corresponding result of Dell, Grohe, and Rattan ($2018$)
    for exact path homomorphism counts to an approximate level.
    Our results answer an open question of Grohe ($2020$) and help to build
    a theoretical understanding of vector embeddings of graphs.

    The distance measures we define turn out be closely related
    to the cut distance.
    We establish our main results by generalizing our definitions to graphons,
    which are limit objects of sequences of graphs,
    as this allows us to apply techniques from functional analysis.
    We prove the fairly general statement that,
    for every \enquote{reasonably}
    defined graphon pseudometric,
    an exact correspondence to homomorphism densities
    can be turned into an approximate one.
    We also provide an example of a distance measure that violates
    this reasonableness condition.
    This incidentally answers an open question of
    Grebík and Rocha ($2021$).
\end{abstract}

\section{Introduction}
\label{sec:introduction}
Vector representations of graphs allow to apply standard machine learning
techniques to graphs, and a variety of methods to generate such embeddings
has been studied in the machine learning literature.
However, from a theoretical point of view, these embeddings have not received
much attention and are not well understood.
Some machine learning methods only implicitly operate on such
vector representations as they only access the inner products
of these vectors.
These methods are known as \textit{kernel methods} and most
graph kernels are based on counting
occurrences of certain substructures, e.g., walks or trees.
See \cite{Grohe2020} for a recent survey  on vector embeddings.

Many kinds of substructure counts in a graph such as graph motifs
are actually just \textit{homomorphism}
counts \enquote{in disguise}, and hence,
homomorphisms provide a formal and flexible framework for counting
all kinds of substructures in graphs \cite{Curticapean2017};
a homomorphism from a graph $F$ to a graph $G$ is a mapping
from the vertices of $F$ to the vertices of $G$ such that every edge
of $F$ is mapped to an edge of $G$.
A theorem of Lovász from 1967 \cite{Lovasz1967}, which states
that two graphs $G$ and $H$ are isomorphic if and only if,
for every graph $F$, the number $\hom(F, G)$ of homomorphisms from $F$ to $G$
equals the corresponding number $\hom(F, H)$ from $F$ to $H$,
led to the development of the theory of graph limits \cite{BorgsEtAl2006GraphLimits, Lovasz2012}, where one considers convergent sequences of graphs and their limit objects,
graphons.
In terms of the \textit{homomorphism vector}
$\Hom(G) \coloneqq (\hom(F, G))_{F \text{ graph}}$
of a graph $G$, the result of Lovász states that graphs are mapped to the same
vector if and only if they are isomorphic.

Computing an entry of $\Hom(G)$
is $\#P$-complete and recent results have
mostly focused on restrictions
$\Hom_\mathcal{F}(G) \coloneqq (\hom(F, G))_{F \in \mathcal{F}}$
of these vectors to classes $\mathcal{F}$
for which computing these entries is actually tractable.
Under a natural assumption from parameterized complexity theory,
this is the case for precisely the classes
$\mathcal{F}$ of bounded tree width \cite{DalmauJonson2004}.
This has led to various surprisingly clean results, e.g.,
for trees and, more general, graphs of bounded treewidth \cite{Dvorak2010},
cycles and paths \cite{Dell2018}, planar graphs \cite{MancinskaEtAl2019},
and, most recently, graphs of bounded tree-depth \cite{Grohe2020TreeDepth}.
These results only show what it means
for graphs to be mapped to the same homomorphism vector;
they do not say anything about the similarity of two graphs
if the homomorphism vectors are not exactly the same but \textit{close}.
Grohe formulated the vague hypothesis that, for suitable classes $\mathcal{F}$,
the embedding $\Hom_\mathcal{F}$ combined with a suitable inner product
on the latent space
induces a natural similarity measure on graphs \cite{Grohe2020}.
This is supported by initial experiments, which show that homomorphism
vectors in combination with support vector machines perform well on
standard graph classification.
Our results further support this hypothesis from a theoretical standpoint
by showing that tree homomorphism counts provide a robust similarity measure.

For the class $\trees$ of trees and two graphs $G$ and $H$, we have
$\Hom_\trees(G) = \Hom_\trees(H)$
if and only if $G$ and $H$ are
not distinguished by \textit{color refinement} (also known as the
\textit{$1$-dimensional Weisfeiler-Leman algorithm}) \cite{Dvorak2010},
a popular heuristic for graph isomorphism.
Another characterization of this equivalence
due to Tinhofer
is that of \textit{fractional isomorphism} \cite{Tinhofer1986}, \cite{Tinhofer1991}.
Let $A \in \R^{V(G) \times V(G)}$
and $B \in \R^{V(H) \times V(H)}$ be the adjacency matrices of $G$ and $H$, respectively,
and consider the following system $\Fiso$ of linear equations:
\begin{equation*}
    \Fiso:
    \begin{cases}
        AX = XB\\
        X \allOne_{V(H)} = \allOne_{V(G)}\\
        \allOne_{V(G)}^TX = \allOne_{V(H)}^T
    \end{cases}
\end{equation*}
Here, $X$ denotes a ($V(G) \times V(H)$)-matrix of variables, and $\allOne_{U}$
denotes the all-$1$ vector over the index set $U$.
The non-negative integer solutions to $\Fiso$
are precisely the permutation matrices that describe isomorphisms between
$G$ and $H$. The non-negative real solutions are called \textit{fractional
isomorphisms} of $G$ and $H$.
Tinhofer proved that $G$ and $H$ are not distinguished by the color refinement
algorithm if and only if there is a fractional isomorphism of $G$ and $H$.
Grohe proposed to define a similarity measure based on this characterization \cite{Grohe2020}:
For a matrix norm $\lVert \cdot \rVert$ that is invariant under permutations
of the rows and columns,
consider
\begin{equation*}
    \groheDist(G, H) \coloneqq \min_{\substack{X \in [0,1]^{V(G) \times V(H)},\\ X \text{ doubly stochastic}}}\lVert AX - XB \rVert.
\end{equation*}
Most graph distance measures based on matrix norms are highly intractable
as the problem of their computation is
related to notoriously hard
maximum quadratic assignment problem \cite{Nagarajan2009}.
This hardness, which stems from the minimization over the set
of all permutation matrices,
motivated Grohe to propose $\groheDist$,
where the set of all permutation matrices is relaxed
to the the convex set of doubly stochastic matrices,
yielding a convex optimization problem.
With the results of Tinhofer and Dvořák, we know that
the graphs of distance zero w.r.t.\ $\groheDist$ are precisely those
that cannot be distinguished by tree homomorphism counts.

So far, the only known connection between a graph distance measure
based on matrix norms and graph homomorphisms
is between the \textit{cut distance} and normalized homomorphism numbers
(called \textit{homomorphism densities}) \cite{BorgsEtAl2006GraphLimits}.
Grohe asks whether a similar correspondence between
$\groheDist$
and restricted homomorphism vectors can be established,
and we give a positive answer to this question.
We introduce the \textit{tree distance} $\treeDist$ of graphs, which is a normalized variant
of $\groheDist$ and show the following theorem,
which is stated here only informally.
We also introduce the \textit{path distance} $\pathDistNoNorm$ of graphs
and prove the analogous theorem to \Cref{th:informal}
for $\pathDistNoNorm$ and normalized path homomorphism counts.
\begin{theorem}[Informal \Cref{th:countingLemmaTreesGraphs} and \Cref{th:inverseCountingLemmaTreesGraphs}]
    \label{th:informal}
    Two graphs $G$ and $H$ are similar w.r.t.\ $\treeDist$
    if and only if
    the homomorphism densities $t(T, G)$ and $t(T, H)$ are close for
    trees $T$.
\end{theorem}

In the theory of graph limits, \textit{graphons} serve as limit objects
for sequences of graphs.
By defining distance measures on the more general graphons,
we are able to use techniques from functional analysis to show that any
\enquote{reasonably} defined pseudometric on graphons satisfying an exact
correspondence to homomorphism densities also has to satisfy
an approximate one.
As an application, we get that both the tree and the path distance
satisfy this correspondence to tree and path homomorphism densities, respectively.
For the case of trees, we rely on
a generalization of the notion of fractional isomorphism
to graphons by Grebík and Rocha \cite{GrebikRocha2019}.
For the case of paths,
we prove this generalization of the result of Dell, Grohe, and Rattan \cite{Dell2018}
by ourselves.

This paper is organized as follows.
In the preliminaries, \Cref{sec:preliminaries},
we collect the definitions of graphs,
the space $\LT$, graphons,
and the cut distance.
In \Cref{sec:pseudometricsForGraphs},
we define the tree distance and the path distance for graphs
and formally state \Cref{th:informal} and its path counterpart.
In \Cref{sec:pseudometrics}, we state and prove
the theorems that allow us to show these correspondences
for graphon pseudometrics.
\Cref{sec:trees} provides the first application of these tools
for the tree distance:
we first state the needed result of
fractional isomorphism of graphons due to Grebík and Rocha
and then use this to define the tree distance of graphons.
These definitions and results specialize to the ones
presented in \Cref{sec:pseudometricsForGraphs} for graphs.
The treatment of the path distance for graphons
is similar to the one of the tree distance,
except for the fact that we prove
a characterization of graphons with the same path homomorphism densities
ourselves,
and can be found in \Cref{sec:paths}.
In \Cref{subsec:cutDistanceInvariant},
we define another distance measure on graphs based on the invariant
computed by the color refinement algorithm
and show that it only satisfies one direction of the approximate
correspondence to tree homomorphism densities.
Our counterexample incidentally
answers an open question of Grebík and Rocha \cite{GrebikRocha2019}.
\Cref{sec:conclusion} poses some interesting open questions
that come up during the study of these distance measures.
All missing proofs are collected in \Cref{sec:appendix}
together with a compilation of results
on operators, graphons, and Markov operators
used in these proofs.

\section{Preliminaries}
\label{sec:preliminaries}

\subsection{Graphs}
\label{subsec:graphs}
By the term graph, we refer to a
simple, undirected, and finite graph.
For a graph $G$, we denote its vertex set by $V(G)$ and its edge set by $E(G)$,
and we let $\vs(G) \coloneqq \lvert V(G) \rvert$ and $\es(G) \coloneqq \lvert E(G) \rvert$.
We usually view the adjacency matrix $A$ of a graph $G$ as a matrix
$A \in \R^{V(G) \times V(G)}$, i.e., it is indexed by the vertices of $G$.
Sometimes, we assume without loss of generality that the vertex set of a graph is
$[n] \coloneqq \{1, \dots, n\}$, where $n \in \N$ is a natural number.
A homomorphism from a graph $F$ to a graph $G$ is a mapping
$\varphi \colon V(F) \to V(G)$ such that
$\varphi(u) \varphi(v) \in E(G)$ for every $uv \in E(F)$.
We denote the number of homomorphisms from $F$ to $G$ by
$\hom(F, G)$.
The homomorphism density from $F$ to $G$ is given by
$t(F, G) \coloneqq {\hom(F, G)}/{\vs(G)^{\vs(F)}}$.

A \textit{weighted graph} $G = (V, a, B)$ consists of a vertex set $V$,
a positive real vector $a = (\alpha_v)_{v \in V} \in \R^V$ of vertex weights
and a real symmetric matrix $B  = (\beta_{uv})\in [0,1]^{V \times V}$ of edge weights;
that is, we restrict ourselves to edge weights from $[0,1]$.
We write $\vs(G) = |V|$, $V(G) = V$, $\alpha_v(G) = \alpha_v$, $\alpha_G = \sum_{v \in V(G)} \alpha_v(G)$ and $\beta_{uv}(G) = \beta_{uv}$.
A weighted graph is called \textit{normalized} if $\alpha_G = 1$.
For a simple graph $F$ and a weighted graph $G$, we define the \textit{homomorphism number}
\begin{equation*}
    \hom(F,G) = \sum_{\varphi \colon V(F) \to V(G)} \prod_{v \in V(F)} \alpha_{\varphi(v)}(G) \prod_{uv \in E(F)} \beta_{\varphi(u)\varphi(v)}(G)
\end{equation*}
and the \textit{homomorphism density}
$t(F,G) = {\hom(F,G)}/{\alpha_G^{\vs(F)}}$.
When viewing a graph as a weighted graph in the obvious way,
these notions coincide with the ones for graphs.

\subsection{The Space \texorpdfstring{$L_2[0,1]$}{L2[0,1]} and Graphons}
\label{subsec:graphons}
A detailed introduction to functional analysis can be found in \cite{Dudley2002};
here, we only repeat some notions we use throughout the main body of the paper.
Let $\LT$ denote the space of $\R$-valued $2$-integrable functions on $[0,1]$
(modulo equality almost anywhere).
We could consider consider an arbitrary standard Borel space instead,
but for the sake of convenience, we stick to $[0,1]$ with the Lebesgue
measure just as \cite{Lovasz2012}.
The space $\LT$ is a Hilbert space
with the inner product defined by
$\langle f, g \rangle \coloneqq \int_{[0,1]} f(x) g(x) \dx$
for functions $f, g \in \LT$.
Let $T \colon \LT \to \LT$ be a bounded linear operator, or operator for short.
We write $\normRTT{T}$ for its operator norm, i.e.,
$\normRTT{T} = \sup_{\lVert g \rVert_2 \le 1} \lVert T g \rVert_2$.
The \textit{Hilbert adjoint} of $T$
is the unique operator $T^* \colon \LT \to \LT$ such that
$\langle T f, g \rangle = \langle f, T^* g \rangle$ for all $f, g \in \LT$,
and $T$ is called self-adjoint if $T^* = T$.

Let $\kernels$ denote the set of all bounded symmetric measurable functions
$W \colon [0,1]^2 \to \R$, called \textit{kernels}.
Let $\graphons \subseteq \kernels$ denote all such $W$ that satisfy
$0 \le W \le 1$;
such a $W$ is called a \textit{graphon}.
Every kernel $W \in \kernels$ defines a self-adjoint
operator $T_W \colon \LT \to \LT$ by
setting
$(T_W f)(x) = \int_{[0,1]} W(x,y) f(y) \dy$
for every $x \in [0,1]$,
which then is a Hilbert-Schmidt operator,
and in particular, compact \cite{Lovasz2012}.

A kernel $W \in \kernels$ is called a \textit{step function} if
there is a partition $S_1 \cup \dots \cup S_k$ of $[0,1]$ such that
$W$ is constant on $S_i \times S_j$ for all $i,j \in \numTo{k}$.
For a weighted graph $H$ on $[n]$,
one can define a step function $W_H \in \kernels$
by splitting $[0,1]$ into $n$ intervals $I_1, \dots, I_n$,
where $I_i$ has length $\lambda(I_i) = {\alpha_i(H)}/{\alpha(H)}$
for every $i \in [n]$,
and letting $W_H(x,y) \coloneqq \beta_{ij}(H)$ for all $x \in I_i, y \in I_j$
and $i,j \in [n]$.
Of course, $W_H$ depends on the labeling of the vertices of $H$.
Note that $W_H$ is a graphon,
and in particular, $W_G$ is a graphon for every graph $G$.

\subsection{The Cut Distance}
\label{subsec:cutdistance}

See \cite{Lovasz2012} for a thorough introduction to the cut distance.
The usual definition of the cut distance involves the \textit{blow-up} $G(k)$
of a graph $G$ by $k \ge 0$, where every vertex of $G$ is replaced
by $k$ identical copies, to get graphs on the same number of vertices.
Going this route is rather cumbersome,
and we directly define the cut distance for weighted graphs
via \textit{fractional overlays};
this definition also applies to graphs in the straightforward way.
A \textit{fractional overlay} of weighted graphs $G$ and $H$
is a matrix $X \in \mathbb{R}^{V(G) \times V(H)}$ such that
    $X_{uv} \ge 0$ for all $u \in V(G)$, $v \in V(H)$,
    $\sum_{v \in V(H)} X_{u v} = {\alpha_u(G)}/{\alpha_G}$ for every $u \in V(G)$, and
    $\sum_{u \in V(G)} X_{u v} = {\alpha_v(H)}/{\alpha_H}$ for every $v \in V(H)$.
Let $\mathcal{X}(G, H)$ denote the set of all fractional overlays of $G$ and $H$.
Note that, for graphs $G$ and $H$,
the second and third condition just say that the row and column
sums of $X$ are ${1}/{\vs(G)}$ and ${1}/{\vs(H)}$, respectively.
For weighted graphs $G$ and $H$ and a fractional overlay $X \in \mathcal{X}(G, H)$, let
    \begin{equation*}
        \fracdist(G, H, X) \coloneqq \max_{Q, R \subseteq V(G) \times V(H)} \Big\lvert \sum_{\substack{iu \in Q,\\ jv \in R}} X_{iu} X_{jv} (\beta_{ij}(G) - \beta_{uv}(H)) \Big\rvert.
    \end{equation*}
Then, define the \textit{cut distance}
$\dist(G, H) \coloneqq \min_{X \in \mathcal{X}(G, H)} \fracdist(G, H, X)$.

Defining the cut distance of graphons is actually much simpler.
Define the \textit{cut norm} on the linear space $\kernels$ of kernels by
$\lVert W \rVert_\square \coloneqq \sup_{S, T \subseteq [0,1]}\left\lvert \int_{S \times T} W(x,y) \,dx \,dy \right\rvert$
for $W \in \kernels$;
here, as in the whole of the paper,
we tacitly assume sets (and functions) we take an infimum or supremum over
to be measurable.
Let $\measPres$ denote the group of all invertible measure-preserving maps
$\varphi \colon [0,1] \to [0,1]$.
For a kernel $W \in \kernels$ and a $\varphi \in \measPres$,
let $W^\varphi$ be the kernel defined by
$W^\varphi(x,y) \coloneqq W(\varphi(x), \varphi(y))$.
For kernels $U, W \in \kernels$, define
their \textit{cut distance} by setting
$\dist(U, W) \coloneqq \inf_{\varphi \in \measPres} \lVert U - W^\varphi \rVert_\square$.
This coincides with the previous definition
when viewing weighted graphs as graphons \cite[Lemma $8.9$]{Lovasz2012}.
We can also express $\dist(U, W)$ via
the kernel operator as
$\dist(U, W) = \inf_{\varphi \in \measPres} \sup_{f,g \colon [0,1] \to [0,1]} \big\lvert \langle f, T_{U - W^\varphi} g \rangle \big\rvert$ \cite[Lemma $8.10$]{Lovasz2012}.
The definition of the cut distance is quite robust.
For example, allowing $f$ and $g$ in the previous definition to be complex-valued
or choosing a different operator norm does not make a difference in most cases \cite[Appendix E]{Janson2013}.

For a graph $F$ and a kernel $W \in \kernels$,
define the \textit{homomorphism density}
\begin{equation*}
    t(F, W) \coloneqq \int_{[0,1]^{V(F)}} \prod_{ij \in E(F)} W(x_i, x_j)\prod_{i \in V(F)} dx_i,
\end{equation*}
which coincides with the previous definition when viewing
weighted graphs as graphons \cite[Equation ($7.2$)]{Lovasz2012}.
\Cref{th:countingLemmaGraphons} and \Cref{th:inverseCountingLemmaGraphons}
state the connection between the cut distance and homomorphism densities:
Informally, the \Cref{th:countingLemmaGraphons} states that graphons that are close in the
cut distance have similar homomorphism densities, while \Cref{th:inverseCountingLemmaGraphons}
states that graphs that have similar homomorphism densities
are close in the cut distance.
We refer to such statements as a \textit{counting lemma} and an \textit{inverse counting lemma},
respectively.
\begin{lemma}[Counting Lemma \cite{LovaszSzegedy2006}]
    \label{th:countingLemmaGraphons}
    Let $F$ be a simple graph, and let $U,W \in \graphons$ be graphons.
    Then,
$\lvert t(F, U) - t(F, W) \rvert \le \es(F) \cdot \dist(U, W)$.
\end{lemma}
\begin{lemma}[Inverse Counting Lemma \cite{BorgsEtAl2008}, \cite{Lovasz2012}]
    \label{th:inverseCountingLemmaGraphons}
    Let $k > 0$, let $U,W \in \graphons$ be graphons,
    and assume that, for every graph $F$ on $k$ vertices,
    we have $\lvert t(F, U) - t(F, W) \rvert \le 2^{-k^2}$.
    Then,
$\dist(U,W) \le {50}/{\sqrt{\log{k}}}$.
\end{lemma}
In particular, graphons $U$ and $W$ have cut distance zero
if and only if, for every graph $F$, we have $t(F, U) = t(F, W)$.
Call a sequence $(W_n)_{n \in \N}$ of graphons \textit{convergent} if,
for every graph $F$, the sequence $(t(F, W_n))_{n \in \N}$ is Cauchy.
The two theorems above yield that $(W_n)_{n \in \N}$ is convergent
if and only if $(W_n)_{n \in \N}$ is Cauchy in $\dist$.
Let $\graphonsWI$ be obtained from $\graphons$ by identifying graphons
with cut distance zero;
such graphons are called \textit{weakly isomorphic}.
One of the main results from graph limit theory is the compactness
of the space $(\graphonsWI, \dist)$.
\begin{theorem}[\cite{LovaszSzegedy2007}]
    \label{th:compact}
    The space $(\graphonsWI,\dist)$ is compact.
\end{theorem}

\section{Similarity Measures of Graphs}
\label{sec:pseudometricsForGraphs}

In this section, we define the tree and path distances
of graphs and formally state the correspondences
to tree and path homomorphism densities, respectively.
All presented results are specializations
of the results for graphons proven in
\Cref{sec:trees} and \Cref{sec:paths}.

\subsection{The Tree Distance of Graphs}
\label{subsec:treeDistanceGraphs}

Recall that two graphs $G$ and $H$ have the same tree homomorphism counts
if and only if the system $\Fiso$ of linear equations has a non-negative solution.
Based on this, Grohe proposed
$\groheDist$
as a similarity measure of graphs.
This is nearly what we define as the tree distance of graphs.
What is missing is, first, a more general definition for graphs with
different numbers of vertices and, second, an appropriate choice
of a matrix norm with an appropriate normalization factor;
analogously to the cut distance, we normalize the tree distance to values in $[0,1]$.
As in the definition of the cut distance in the preliminaries,
we handle graphs on different numbers of vertices by
considering fractional overlays instead of blow-ups (and doubly stochastic matrices).
Recall that a fractional overlay of graphs $G$ and $H$
is a matrix $X \in \mathbb{R}^{V(G) \times V(H)}$ such that
    $X_{uv} \ge 0$ for all $u \in V(G)$, $v \in V(H)$,
    $\sum_{v \in V(H)} X_{u v} = {1}/{\vs(G)}$ for every $u \in V(G)$, and
    $\sum_{u \in V(G)} X_{u v} = {1}/{\vs(H)}$ for every $v \in V(H)$.
If $\vs(G) = \vs(H)$, then the difference between
a fractional overlay and a doubly stochastic matrix is just a factor of $\vs(G)$.
Also recall that $\mathcal{X}(G, H)$ denotes the set of all fractional overlays of $G$ and $H$.

We consider two matrix norms for the tree distance:
First, just like in the definition of the cut distance,
we use the \textit{cut norm} for matrices,
introduced by Frieze and Kannan \cite{FriezeKannan1999},
defined as $\normC{A} \coloneqq \max_{S \subseteq [m], T \subseteq [n]} \lvert \sum_{i \in S, j \in T} A_{ij} \rvert$
for $A \in \R^{m \times n}$.
Second, we also consider the more standard spectral norm
$\normT{A} \coloneqq \sup_{x \in \R^n, \normT{x} \le 1} \normT{Ax}$
of a matrix $A \in \R^{m \times n}$.
From a computational point of view,
the Frobenius norm might also be appealing,
but this would lead to a different topology,
cf.\ \cite[Appendix E]{Janson2013}.
\begin{definition}[Tree Distance of Graphs]
    \label{def:treeDistanceGraphs}
    Let $G$ and $H$ be graphs with adjacency matrices $A \in \R^{V(G) \times V(G)}$
    and $B \in \R^{V(H) \times V(H)}$, respectively.
    Then, define
    \begin{align*}
        \neiDistance(G, H) &\coloneqq \inf_{X \in \mathcal{X}(G, H)} \frac{1}{\vs(G) \cdot \vs(H)} \lVert \vs(H) \cdot AX - \vs(G) \cdot XB \rVert_{\square} \text{ and}\\\neiDistSpec(G, H) &\coloneqq \inf_{X \in \mathcal{X}(G, H)} \frac{1}{\sqrt{\vs(G) \vs(H)}} \normT{\vs(H) \cdot AX - \vs(G) \cdot XB}.
    \end{align*}
\end{definition}
Note that the spectral norm requires an adapted normalization factor
in \Cref{def:treeDistanceGraphs}.
The advantage of $\neiDistance$ is the close connection to the cut distance,
which also utilizes the cut norm.
However, the crucial advantage of the spectral norm is that
minimization of the spectral norm of a matrix is a standard application
of interior-point methods in convex optimization.
In particular, an $\epsilon$-solution to $\neiDistSpec$ can be
computed in polynomial time \cite[Section $6.3.3$]{NesterovNemirovskii1994}.
For $\neiDistance$, it is not clear whether this is possible.

From the results of \Cref{sec:trees}, we get that
$\neiDist$ and $\neiDistSpec$ are pseudometrics (\Cref{le:neiDistPseudoMetric})
and that two graphs have distance zero if and only if their tree homomorphism
densities are the same (\Cref{le:neiDistZero}).
Moreover, we have $\neiDist \le \dist$ (\Cref{le:neiDistLeCutDist}),
and these
pseudometrics are invariant under blow-ups.
Finally, we get the following counting lemma (\Cref{co:countingNeighborhood})
and inverse counting lemma (\Cref{co:invCountingNeighborhood}).

\begin{theorem}[Counting Lemma for $\treeDist$, Graphs]
    \label{th:countingLemmaTreesGraphs}
    Let $\neiDistNo \in \{\neiDist, \neiDistSpec\}$.
    For every tree $T$ and every $\epsilon > 0$, there is an $\eta > 0$ such that,
    for all graphs $G$ and $H$, if
    $\neiDistNo(G, H) \le \eta$,
    then
    $\lvert t(T, G) - t(T, H) \rvert \le \epsilon$.
\end{theorem}

\begin{theorem}[Inverse Counting Lemma for $\treeDist$, Graphs]
    \label{th:inverseCountingLemmaTreesGraphs}
    Let $\neiDistNo \in \{\neiDist, \neiDistSpec\}$.
    For every $\epsilon > 0$, there are $k > 0$ and $\eta > 0$ such that,
    for all graphs $G$ and $H$,
    if
    $\lvert t(T, G) - t(T, H) \rvert \le \eta$
    for every tree $T$ on at most $k$ vertices, then
    $\neiDistNo(G, H) \le \epsilon$.
\end{theorem}

\subsection{The Path Distance of Graphs}
\label{subsec:pathDistanceGraphs}

Dell, Grohe, and Rattan proved that two graphs $G$ and $H$
have the same path homomorphism counts if and only if
the system $\Fiso$ of linear equations has a real solution \cite{Dell2018}.
This transfers to the definition of the path distance, i.e.,
we define the path distance analogously to the tree distance
but relax the non-negativity condition of fractional overlays.
For graphs $G$ and $H$, we call
a matrix $X \in \mathbb{R}^{V(G) \times V(H)}$ a
\textit{signed fractional overlay} of $G$ and $H$ if
    $\normT{X y} \le \normT{y}/\sqrt{\vs(G) \vs(H)}$ for every $y \in \R^{V(H)}$,
    $\sum_{v \in V(H)} X_{u v} = {1}/{\vs(G)}$ for every $u \in V(G)$, and
    $\sum_{u \in V(G)} X_{u v} = {1}/{\vs(H)}$ for every $v \in V(H)$.
Let $\mathcal{S}(G, H)$ denote the set of all signed fractional overlays of $G$ and $H$.
The first condition requires that $X$ is a contraction (up to a scaling factor)
in the spectral norm;
we need this to guarantee that our definition of the path distance
actually yields a pseudometric.
This restriction to the spectral norm
stems from the fact that
the proof of Dell, Grohe, and Rattan \cite{Dell2018}
(and our generalization thereof to graphons)
only guarantees that the constructed solution
is a contraction in the spectral norm,
cf.\ \Cref{sec:paths} for the details.

\begin{definition}[Path Distance of Graphs]
    \label{def:pathDistanceGraphs}
    Let $G$ and $H$ be graphs with adjacency matrices $A \in \R^{V(G) \times V(G)}$
    and $B \in \R^{V(H) \times V(H)}$, respectively.
    Then, define
    \begin{align*}
        \pathDistSpec(G, H) &\coloneqq \inf_{X \in \mathcal{S}(G, H)} \frac{1}{\sqrt{\vs(G) \vs(H)}} \normT{\vs(H) \cdot AX - \vs(G) \cdot XB}.
    \end{align*}
\end{definition}
From \Cref{sec:paths}, we get that
$\pathDistSpec$ is a pseudometric (\Cref{le:pathDistTPseudoMetric}) that is invariant under blow-ups and that
has as graphs of distance zero precisely these
with the same path homomorphism densities.
Moreover,
we get the following (quantitative) counting lemma (\Cref{th:countingLemmaPathDistT})
and inverse counting lemma (\Cref{co:invCountingPathT}).

\begin{theorem}[Counting Lemma for $\pathDistSpec$, Graphs]
    \label{th:countingLemmaPathDistSpec}
    Let $P$ be a path, and let $G$ and $H$ be graphs.
    Then, $\lvert t(P, G) - t(P, H) \rvert \le \es(P) \cdot \pathDistSpec(G, H)$.
\end{theorem}

\begin{theorem}[Inverse Counting Lemma for $\pathDistSpec$, Graphs]
    For every $\epsilon > 0$, there are $k > 0$ and $\eta > 0$ such that,
    for all graphs $G$ and $H$,
    if
    $\lvert t(P, G) - t(P, H) \rvert \le \eta$
    for every path $P$ on at most $k$ vertices, then
    $\pathDistSpec(G, H) \le \epsilon$.
\end{theorem}

\section{Graphon Pseudometrics and Homomorphism Densities}
\label{sec:pseudometrics}

In this section, we provide the main tools
we need to prove the correspondences between
the tree and path distances
and tree and path homomorphism densities, respectively.
Consider a pseudometric $\delta$ on graphons.
We say that \textit{$\delta$ is compatible with $\dist$}
if, for every sequence of graphons $(U_n)_n,$ $U_n \in \graphons$,
and every graphon $\widetilde{U} \in \graphons$,
$\dist(U_n, \widetilde{U}) \xrightarrow{n \rightarrow \infty} 0$
implies $\delta(U_n, \widetilde{U}) \xrightarrow{n \rightarrow \infty} 0$.
For example, this is the case if $\delta \le \dist$, i.e.,
graphons only get closer if we consider $\delta$ instead of $\dist$.
We anticipate that the pseudometrics we are interested in,
the tree distance and the path distance,
are compatible with $\dist$.

Together, the next two theorems state that
every pseudometric that is compatible with $\dist$ and whose
graphons of distance zero can be characterized by homomorphism densities
from a class of graphs $\mathcal{F}$
already has to satisfy both a counting lemma and an inverse counting lemma
for this class $\mathcal{F}$.
The proof of these theorems is a simple compactness argument,
utilizing the compactness of the graphon space, \Cref{th:compact},
and the counting lemma for $\dist$, \Cref{th:countingLemmaGraphons}.
Therefore, it is absolutely crucial
that we consider a pseudometric defined on graphons
as the limit of a sequence of graphs may not be
a graph.

\begin{theorem}[Counting Lemma for $\someClass$]
    \label{th:countingTrees}
    Let $\someClass$ be a class of graphs, and
    let $\someDist$ be a pseudometric on graphons such that
        (1) $\someDist$ is compatible with $\dist$ and (2),
        for all graphons $U, W \in \graphons$, $\someDist(U, W) = 0$ implies
              $t(F, U) = t(F, W)$ for every graph $F \in \someClass$.
    Then, for every graph $F \in \someClass$ and every $\epsilon > 0$, there is an $\eta > 0$ such that,
    for all graphons $U, W \in \graphons$, if
$\someDist(U, W) \le \eta$,
then
$\lvert t(F, U) - t(F, W) \rvert \le \epsilon$.
\end{theorem}
\begin{proof}[Proof of \Cref{th:countingTrees}]
    We proceed by contradiction and assume that the statement does not hold.
    Then, there is a graph $F \in \someClass$ and an $\epsilon > 0$ such that, for every
    $\eta > 0$, there are graphons $U, W \in \graphons$ such that
    $\someDist(U, W) \le \eta$ and $\lvert t(F, U) - t(F, W) \rvert > \epsilon$.

    Let $k > 0$.
    Then, by choosing $\eta = \frac{1}{k}$, we get that there are graphons
    $U_k, W_k \in \graphons$ such that
    $\someDist(U_k, W_k) \le \frac{1}{k}$ and
    $\lvert t(F, U_k) - t(F, W_k) \rvert > \epsilon$.
    By the compactness theorem, \Cref{th:compact}, we get that
    the sequence $(U_k)_{k}$ has a convergent
    subsequence $(U_{k_i})_{i}$ converging to a graphon $\widetilde{U}$
    in the metric $\dist$.
    By another application of that theorem, we get that $(W_{k_i})_{i}$
    has a convergent subsequence $(W_{\ell_i})_{i}$
    converging to a graphon $\widetilde{W}$ in the metric $\dist$.
    Then, $(U_{\ell_i})_i$ and $(W_{\ell_i})_i$ are sequences
    converging to $\widetilde{U}$ and $\widetilde{W}$ in the metric $\dist$,
    respectively.

    Now, for every $i > 0$, we have
    \begin{equation*}
        \someDist(\widetilde{U}, \widetilde{W}) \le \someDist(\widetilde{U}, U_{\ell_i})+ \someDist(U_{\ell_i}, W_{\ell_i})+ \someDist(W_{\ell_i}, \widetilde{W}).
    \end{equation*}
    By assumption, we have
    $\someDist(U_{\ell_i}, W_{\ell_i}) \le \frac{1}{\ell_i}$,
    which means that
    $\someDist(U_{\ell_i}, W_{\ell_i}) \xrightarrow{i \rightarrow \infty} 0$.
    Since $\dist(U_{\ell_i}, \widetilde{U}) \xrightarrow{i \rightarrow \infty} 0$
    and $\dist(W_{\ell_i}, \widetilde{W}) \xrightarrow{i \rightarrow \infty} 0$,
    the first assumption about $\someDist$ yields that
    we also have
    $\someDist(U_{\ell_i}, \widetilde{U}) \xrightarrow{i \rightarrow \infty} 0$
    and $\someDist(W_{\ell_i}, \widetilde{W}) \xrightarrow{i \rightarrow \infty} 0$.
    Hence, we must have $\someDist(\widetilde{U}, \widetilde{W}) = 0$.

    Since $\someDist(\widetilde{U}, \widetilde{W}) = 0$,
    we have $t(F, \widetilde{U}) = t(F, \widetilde{W})$ by
    the second assumption about $\someDist$.
    By the Counting Lemma, \Cref{th:countingLemmaGraphons},
    we get that
    $\lvert t(F, U_{\ell_i}) - t(F, \widetilde{U}) \rvert \xrightarrow{i \rightarrow \infty} 0$
    and
    $\lvert t(F, \widetilde{W}) - t(F, W_{\ell_i}) \rvert \xrightarrow{i \rightarrow \infty} 0$.
    Now, for every $i > 0$, we have
    \begin{equation*}
        \lvert t(F, U_{\ell_i}) - t(F, W_{\ell_i}) \rvert \le \lvert t(F, U_{\ell_i}) - t(F, \widetilde{U}) \rvert
            + \lvert t(F, \widetilde{U}) - t(F, \widetilde{W}) \rvert
            + \lvert t(F, \widetilde{W}) - t(F, W_{\ell_i}) \rvert
    \end{equation*}
    Hence, $\lvert t(F, U_{\ell_i}) - t(F, W_{\ell_i}) \rvert \xrightarrow{i \rightarrow \infty} 0$.
    This contradicts the fact that $\lvert t(F, U_{\ell_i}) - t(F, W_{\ell_i}) \rvert > \epsilon$
    for every $i$.
\end{proof}

Just as the proof of \Cref{th:countingTrees},
the proof of \Cref{th:invCountingTrees}
only relies on the compactness of the graphon space
and the counting lemma for $\dist$, and not on a counting lemma for a
specific class of graphs or the inverse counting lemma for $\dist$.
\begin{theorem}[Inverse Counting Lemma for $\someClass$]
    \label{th:invCountingTrees}
    Let $\someClass$ be a class of graphs, and
    let $\someDist$ be a pseudometric on graphons such that
        (1) $\someDist$ is compatible with $\dist$ and (2),
        for all graphons $U, W \in \graphons$, $t(F, U) = t(F, W)$
              for every graph $F \in \someClass$ implies $\someDist(U, W) = 0$.
    Then, for every $\epsilon > 0$, there are $k > 0$ and $\eta > 0$ such that,
    for all graphons $U, W \in \graphons$,
    if
$\lvert t(F, U) - t(F, W) \rvert \le \eta$
for every graph $F \in \someClass$ on at most $k$ vertices, then
$\someDist(U, W) \le \epsilon$.
\end{theorem}
\begin{proof}
    We proceed by contradiction and assume that the statement does not hold.
    Then, there is an $\epsilon > 0$ such that, for every $k > 0$
    and every $\eta > 0$, there are graphons $U, W \in \graphons$ such that
    $\lvert t(F, U) - t(F, W) \rvert \le \eta$ for every graph $F \in \someClass$
    on at most $k$ vertices
    but $\someDist(U, W) > \epsilon$.

    Let $k > 0$.
    Then, by choosing $\eta = \frac{1}{k}$, we get that there are graphons
    $U_k, W_k \in \graphons$ such that
    $\lvert t(F, U_k) - t(F, W_k) \rvert \le \frac{1}{k}$ for every graph $F \in \someClass$
    on at most $k$ vertices and $\someDist(U_k, W_k) > \epsilon$.
    By the compactness theorem, \Cref{th:compact}, we get that
    the sequence $(U_k)_{k}$ has a convergent
    subsequence $(U_{k_i})_{i}$ converging to a graphon $\widetilde{U}$
    in the metric $\dist$.
    By another application of that theorem, we get that $(W_{k_i})_{i}$
    has a convergent subsequence $(W_{\ell_i})_{i}$
    converging to a graphon $\widetilde{W}$ in the metric $\dist$.
    Then, $(U_{\ell_i})_i$ and $(W_{\ell_i})_i$ are sequences
    converging to $\widetilde{U}$ and $\widetilde{W}$ in the metric $\dist$,
    respectively.

    Let $F \in \someClass$ be a graph.
    Now, for every $i > 0$, we have
    \begin{equation*}
        \lvert t(F, \widetilde{U}) - t(F, \widetilde{W}) \rvert \le \lvert t(F, \widetilde{U}) - t(F, U_{\ell_i}) \rvert + \lvert t(F, U_{\ell_i}) - t(F, W_{\ell_i}) \rvert + \lvert t(F, W_{\ell_i}) - t(F, \widetilde{W})\rvert
    \end{equation*}
    By the counting lemma for $\dist$, \Cref{th:countingLemmaGraphons},
    we get that
    $\lvert t(F, \widetilde{U}) - t(F, U_{\ell_i}) \rvert \xrightarrow{i \rightarrow \infty} 0$
    and
    $\lvert t(F, W_{\ell_i}) - t(F, \widetilde{W})\rvert \xrightarrow{i \rightarrow \infty} 0$.
    Moreover, by assumption, we have
    $\lvert t(F, U_{\ell_i}) - t(F, W_{\ell_i}) \rvert \le \frac{1}{\ell_i}$
    for large enough $i$, which means that also
    $\lvert t(F, U_{\ell_i}) - t(F, W_{\ell_i}) \rvert \xrightarrow{i \rightarrow \infty} 0$.
    Hence, we must have $t(F, \widetilde{U}) = t(F, \widetilde{W})$.

    As we have $t(F, \widetilde{U}) = t(F, \widetilde{W})$ for every graph $F \in \someClass$,
    the second assumption about $\someDist$ yields that
    $\someDist(\widetilde{U}, \widetilde{W}) = 0$.
    Since $\dist(U_{\ell_i}, \widetilde{U}) \xrightarrow{i \rightarrow \infty} 0$
    and $\dist(W_{\ell_i}, \widetilde{W}) \xrightarrow{i \rightarrow \infty} 0$,
    we also have
    $\someDist(U_{\ell_i}, \widetilde{U}) \xrightarrow{i \rightarrow \infty} 0$
    and $\someDist(W_{\ell_i}, \widetilde{W}) \xrightarrow{i \rightarrow \infty} 0$
    by the first assumption about $\someDist$.
    Now, for every $i > 0$, we have
    \begin{equation*}
        \someDist(U_{\ell_i}, W_{\ell_i}) \le \someDist(U_{\ell_i}, \widetilde{U})+ \someDist(\widetilde{U}, \widetilde{W})+ \someDist(\widetilde{W}, W_{\ell_i}).
    \end{equation*}
    Hence, $\someDist(U_{\ell_i}, W_{\ell_i}) \xrightarrow{i \rightarrow \infty} 0$.
    This contradicts the fact that $\someDist(U_{\ell_i}, W_{\ell_i}) > \epsilon$
    for every $i$.
\end{proof}
 
\section{Homomorphisms from Trees}
\label{sec:trees}

In this section, we define the tree distance of graphons.
To use the results from \Cref{sec:pseudometrics},
we prove that the graphons of distance zero
are precisely those with the same tree homomorphism densities
(\Cref{le:neiDistZero})
and that the tree distance is compatible with the cut distance
(\Cref{le:neiDistLeCutDist}).
As for graphs,
we define two variants of the tree distance,
which yield the same topology (\Cref{le:treeDistInequalitiesSimple}):
one using the analogue of the cut norm
and one using the analogue of the spectral norm.

\subsection{Fractional Isomorphism of Graphons}
\label{subsec:fracIsoGraphons}

Recall that two graphs $G$ and $H$ with adjacency matrices $A \in \R^{V(G) \times V(G)}$
and $B \in \R^{V(H) \times V(H)}$,
respectively, are called fractionally isomorphic
if there is a doubly stochastic matrix $X \in \R^{V(G) \times V(H)}$
such that $AX = XB$.
Grebík and Rocha proved \Cref{th:fracIsoGraphons},
which generalizes this to graphons \cite{GrebikRocha2019};
doubly stochastic matrices become \textit{Markov operators}
\cite{EisnerEtAl2015}.
An operator $S \colon \LT \to \LT$ is called a Markov operator if
$S \ge 0$, i.e., $f \ge 0$ implies $S(f) \ge 0$,
$S(\allOne) = \allOne$, and
$S^*(\allOne) = \allOne$,
where $\allOne$ is the all-one function on $[0,1]$.
We denote the set of all Markov operators $S \colon \LT \to \LT$
by $\markov$.

\begin{theorem}[\cite{GrebikRocha2019}, Part of Theorem $1.2$]
    \label{th:fracIsoGraphons}
    Let $U, W \in \graphons$ be graphons.
    There is a Markov operator $S \colon \LT \to \LT$ such that $T_U \circ S = S \circ T_W$
    if and only if $t(T, U) = t(T, W)$ for every tree $T$.
\end{theorem}

\subsection{The Tree Distance}
\label{subsec:treeDistance}
Recall that, for graphons $U, W \in \graphons$, the cut distance of $U$ and $W$
can be written as
$\dist(U, W) = \inf_{\varphi \in \measPres} \sup_{f,g \colon [0,1] \to [0,1]} \big\lvert \langle f, T_{U - W^\varphi} g \rangle \big\rvert$.
We obtain the tree distance of $U$ and $W$
by relaxing measure-preserving maps to Markov operators.

\begin{definition}[Tree Distance]
    \label{def:neiDistGraphons}
    Let $U,W \in \graphons$ be graphons.
    Then, define
    \begin{align*}
        \neiDist(U, W) &\coloneqq \inf_{S \in \markov} \sup_{f, g \colon [0,1] \to [0,1]} \lvert \langle f, (T_U \circ S - S \circ T_W) g \rangle \rvert \text{ and}\\
        \neiDistROp(U, W) &\coloneqq \inf_{S \in \markov} \normRTT{T_U \circ S - S \circ T_W}.
    \end{align*}
\end{definition}

As the notation $\neiDist$ indicates, the definition of $\neiDist$ is based (although not explicitly) on the cut norm,
while $\neiDistROp$ is defined via the operator norm $\normRTT{\cdot}$,
which corresponds to the spectral norm for matrices.
One can verify that these definitions specialize
to the ones for graphs from \Cref{subsec:treeDistanceGraphs}.
The proof can be found in \Cref{subsec:distancesCoincide}.
\begin{lemma}
    \label{le:steppingFunctionNeiDist}
    \label{le:steppingFunctionNeiDistSpec}
    Let $G$ and $H$ be graphs.
    Then,
    $\neiDist(G, H) = \neiDist(W_G, W_H)$
    and $\neiDistSpec(G, H) = \neiDistROp(W_G, W_H)$.
\end{lemma}

We verify that the tree distance actually is
a pseudometric.
To prove the triangle inequality for $\neiDist$ and $\neiDistROp$,
we use that a Markov operator is a contraction on $\LI$ and $\LT$,
respectively \cite[Theorem $13.2$ b)]{EisnerEtAl2015}.
The proof can be found in \Cref{subsec:missingPseudoMetricProof}.

\begin{lemma}
    \label{le:neiDistPseudoMetric}
    $\neiDist$  and $\neiDistROp$ are pseudometrics on $\graphons$.
\end{lemma}

The Riesz-Thorin Interpolation Theorem
(see, e.g., \cite[Theorem $1.1.1$]{Bergh1976})
allows to prove that
both variants of the tree distance
define the same topology.
The proof of \Cref{le:treeDistInequalitiesSimple} can be found in \Cref{subsec:treeDistanceNorms}.

\begin{lemma}
    \label{le:treeDistInequalitiesSimple}
    Let $U, W \in \graphons$ be graphons.
    Then,
    $\neiDist(U, W) \le \neiDistROp(U, W) \le 4 \neiDist(U, W)^{1/2}$.
\end{lemma}

To be able to apply the results from \Cref{sec:pseudometrics},
we need that the tree distance of two graphons is zero
if and only if their tree homomorphism densities are the same.
Let $U, W \in \graphons$ be graphons.
From the respective definitions, it is not immediately
clear that $\neiDist(U, W) = 0$ or $\neiDistROp(U, W) = 0$
implies $t(T, U) = t(T, W)$ for every tree $T$ since
the infimum over all Markov operators might not be attained.
Here, we can use a continuity argument
as the set of Markov operators is compact
in the weak operator topology \cite[Theorem $13.8$]{EisnerEtAl2015}.
However, we have to take a detour via a third variant
of the tree distance where
compactness in the weak operator topology suffices.
All the details can be found in \Cref{subsec:infimumAttained}.

\begin{lemma}
    \label{le:neiDistZero}
    Let $U, W \in \graphons$ be graphons.
    Then, $\neiDist(U, W) = 0$ if and only if $t(T,U) = t(T,W)$ for every tree $T$.
\end{lemma}

The \textit{Koopman operator} $T_\varphi \colon f \mapsto f \circ \varphi$
of a measure-preserving map $\varphi \colon [0,1] \to [0,1]$
is a Markov operator \cite[Example $13.1$, 3)]{EisnerEtAl2015}.
Hence, the tree distance can be seen as the relaxation of the cut distance
obtained by relaxing measure-preserving maps to Markov operators.
In particular, this means that the tree distance is compatible
with the cut distance.
The proof of \Cref{le:neiDistLeCutDist} can be found in
\Cref{subsec:treeDistLeCutDist}.

\begin{lemma}
    \label{le:neiDistLeCutDist}
    \label{le:neiDistCutDistConvergence}
    Let $U, W \in \graphons$ be graphons.
    Then,
    $\neiDist(U, W) \le \dist(U, W)$.
\end{lemma}

With \Cref{le:neiDistZero} and \Cref{le:neiDistCutDistConvergence}
we can apply the theorems of \Cref{sec:pseudometrics} and get both a counting
lemma and an inverse counting lemma for the tree distance.

\begin{corollary}[Counting Lemma for $\treeDist$]
    \label{co:countingNeighborhood}
    Let $\neiDistNo \in \{\neiDist, \neiDistROp\}$.
    For every tree $T$ and every $\epsilon > 0$, there is an $\eta > 0$ such that,
    for all graphons $U, W \in \graphons$, if
    $\neiDistNo(U, W) \le \eta$,
    then
    $\lvert t(T, U) - t(T, W) \rvert \le \epsilon$.
\end{corollary}

\begin{corollary}[Inverse Counting Lemma for $\treeDist$]
    \label{co:invCountingNeighborhood}
    Let $\neiDistNo \in \{\neiDist, \neiDistROp\}$.
    For every $\epsilon > 0$, there are $k > 0$ and $\eta > 0$ such that,
    for all graphons $U, W \in \graphons$,
    if
    $\lvert t(T, U) - t(T, W) \rvert \le \eta$
    for every tree $T$ on $k$ vertices, then
    $\neiDistNo(U, W) \le \epsilon$.
\end{corollary}

\section{Homomorphisms from Paths}
\label{sec:paths}

In this section, we define the path distance of graphons.
We prove a quantitative counting lemma for it (\Cref{th:countingLemmaPathDistT})
and only rely on the results from \Cref{sec:pseudometrics}
to obtain an inverse counting lemma.
To this end, we
we prove that the graphons of distance zero
are precisely those with the same path homomorphism densities
(\Cref{le:pathDistTZero})
and that the path distance is compatible with the cut distance
(\Cref{le:pathDistTLeNeiDistT}).
Since there is no existing characterization of graphons
with the same path homomorphism densities that we can rely on,
we first generalize the result of Dell, Grohe, and Rattan
to graphons (\Cref{th:pathHomsGraphons}).

\subsection{Path Densities and Graphons}
\label{subsec:pathsGraphons}
Dell, Grohe, and Rattan have shown the surprising fact that $G$ and $H$
have the same path homomorphism counts if and only if the system $\Fiso$
has a real solution \cite{Dell2018}.
We need a generalization of their characterization to graphons
in order to define the path distance of graphons and
apply the results from \Cref{sec:pseudometrics}.
If two graphons $U, W \in \graphons$ have the same path homomorphism
densities, the proof of \Cref{th:pathHomsGraphons} yields an operator
$S \colon \LT \to \LT$ such that $S(\allOne) = \allOne$ and $S^*(\allOne) = \allOne$,
which generalizes the result of \cite{Dell2018} in a straight-forward fashion.
An important detail is that the proof also yields that $S$ is an $L_2$-contraction;
this guarantees that the path distance
satisfies the triangle inequality, i.e.,
that it is a pseudometric in the first place.
For the sake of brevity,
we call an operator $S \colon \LT \to \LT$ a \textit{signed Markov operator} if
$S$ is an $L_2$-contraction, i.e., $\normT{S f} \le \normT{f}$ for every $f \in \LT$,
$S(\allOne) = \allOne$, and
$S^*(\allOne) = \allOne$.
Let $\signedMarkovs$ denote the set of all signed Markov operators.
It is easy to see that $\signedMarkovs$ is closed under composition and
Hilbert adjoints.
\begin{theorem}
    \label{th:pathHomsGraphons}
    Let $U, W \in \graphons$.
    There is a signed Markov operator $S \colon \LT \to \LT$
    such that $T_U \circ S = S \circ T_W$ if and only if
    $t(P, U) = t(P, W)$ for every path $P$.
\end{theorem}

Homomorphism densities from paths can be expressed in terms of operator powers.
For $\ell \ge 0$, let $P_\ell$ denote the path of length $\ell$.
Then, for a graphon $U$, we have
\begin{align*}
    t(P_\ell, U) = \int_{[0,1]^{\ell + 1}} \prod_{i \in [\ell]} U(x_{i}, x_{i+1}) \prod_{i \in [\ell + 1]} \dx_i = \langle \allOne, T_U^{\ell} \allOne \rangle
\end{align*}
for every $\ell \ge 0$.
The proof of \Cref{th:pathHomsGraphons}
utilizes the Spectral Theorem
for compact operators on Hilbert spaces
to express $\allOne$ as a sum of orthogonal eigenfunctions.
For a kernel $W \in \kernels$, $T_W \colon \LT \to \LT$ is a Hilbert-Schmidt
operator and, hence, compact \cite{Lovasz2012}.
Since $\LT$ is separable and $T_W$ is compact and self-adjoint,
the Spectral Theorem yields that there is a
countably infinite orthonormal basis $\{f'_i\}$ of $\LT$ consisting of eigenfunctions
of $T_W$ with the corresponding multiset of eigenvalues $\{\lambda_n\} \subseteq \R$ such that
$\lambda_n \xrightarrow{n \to \infty} 0$ (see, e.g., \cite{DunfordSchwartz1988}).
If graphons $U$ and $W$ have the same path homomorphism densities,
an interpolation argument yields that
the lengths of the eigenvectors in the decomposition of $\allOne$
and their eigenvalues have to be the same.
Then, one can define the operator $S$ from these eigenfunctions of $U$ and $W$.
The detailed proof can be found in \Cref{subsec:pathDensitiesProof}.

\subsection{The Path Distance}
\label{subsec:pathDistance}

We define the \textit{path distance} of graphons
can analogously to the tree distance.
However, as the proof \Cref{th:pathHomsGraphons} does not yield that
the resulting operator is an $L_\infty$-contraction, we are limited
in our choice of norms.
\begin{definition}[Path Distance]
    \label{def:pathDistTGraphons}
    Let $U,W \in \graphons$ be graphons.
    Then, define
    \begin{equation*}
        \pathDistROp(U, W) \coloneqq \inf_{S \in \signedMarkovs}  \normRTT{T_U \circ S - S \circ T_W}.
    \end{equation*}
\end{definition}

One can verify that this definition specializes
to the one for graphs from \Cref{subsec:pathDistanceGraphs}.
The proof can be found in \Cref{subsec:distancesCoincide}.
\begin{lemma}
    \label{le:steppingFunctionPathDistSpec}
    Let $G$ and $H$ be graphs.
    Then, $\pathDistSpec(G, H) = \pathDistROp(W_G, W_H)$.
\end{lemma}

The proof that $\pathDistROp$ is a pseudometric
can be found in \Cref{subsec:missingPseudoMetricProof}.
\begin{lemma}
    \label{le:pathDistTPseudoMetric}
    $\pathDistROp$ is a pseudometric on $\graphons$.
\end{lemma}

To apply the theorems of \Cref{sec:pseudometrics},
we need that two graphons have distance zero in the path distance if and only if
their path homomorphism densities are the same and that $\pathDistROp$
is compatible with $\dist$.
For the former, we deviate from the way we proceeded for the tree distance
as we actually can prove a quantitative counting lemma.

\begin{theorem}[Counting Lemma for Paths]
    \label{th:countingLemmaPathsGraphons}
    Let $P$ be a path, and let $U,W \in \graphons$ be graphons.
    Then, for every operator $S \colon \LT \to \LT$ with $S(\allOne) = \allOne$
    and $S^*(\allOne) = \allOne$,
    \begin{equation*}
        \lvert t(P, U) - t(P, W) \rvert \le \es(P) \cdot \sup_{f, g \colon [0,1] \to [0,1]} \lvert \langle f, (T_U \circ S - S \circ T_W) g \rangle \rvert.
    \end{equation*}
\end{theorem}
\begin{proof}
    Let $\ell \in \N$ and $S \in \signedMarkovs$.
    Then,
    \begin{align*}
        \lvert t(P_\ell, U) - t(P_\ell, W) \rvert &= \lvert \langle \allOne, T_U^\ell (S \allOne) \rangle - \langle (S^* \allOne), T_W^\ell \allOne \rangle \rvert\\
&= \big\lvert \sum_{i \in [\ell]} \left( \langle \allOne, (T_U^{\ell - i + 1} \circ S \circ T_W^{i - 1}) \allOne \rangle - \langle \allOne, (T_U^{\ell - i} \circ S \circ T_W^{i}) \allOne \rangle \right) \big\rvert\\
&= \big\lvert \sum_{i \in [\ell]} \langle T_U^{\ell - 1} \allOne, (T_U \circ S - S \circ T_W) (T_W^{i - 1} \allOne) \rangle \big\rvert\\
&\le \ell \cdot \sup_{f, g \colon [0,1] \to [0,1]} \lvert \langle f, (T_U \circ S - S \circ T_W) g \rangle \rvert.
    \end{align*}
\end{proof}

\Cref{th:countingLemmaPathsGraphons} suggests that,
for graphons $U,W \in \graphons$, one should define
\begin{equation*}
    \pathDist(U, W) \coloneqq \inf_{S \in \signedMarkovs} \sup_{f, g \colon [0,1] \to [0,1]} \lvert \langle f, (T_U \circ S - S \circ T_W) g \rangle \rvert.
\end{equation*}
Then, we have
$\lvert t(P, U) - t(P, W) \rvert \le \es(P) \cdot \pathDist(U, W)$
for every path $P$.
However, as mentioned before, we cannot verify
that $\pathDist$ is a pseudometric as the operator $S$ might not be an
$L_\infty$-contraction.

\begin{corollary}[Counting Lemma for $\pathDistROp$]
    \label{th:countingLemmaPathDistT}
    Let $P$ be a path, and let $U,W \in \graphons$ be graphons.
    Then,
        $\lvert t(P, U) - t(P, W) \rvert \le \es(P) \cdot \pathDistROp(U, W)$.
\end{corollary}
\begin{proof}
    By the Cauchy-Schwarz inequality, we have
    \begin{align*}
        \sup_{f, g \colon [0,1] \to [0,1]} \lvert \langle f, (T_U \circ S - S \circ T_W) g \rangle \rvert &\le \sup_{f, g \colon [0,1] \to [0,1]} \normT{f} \normT{(T_U \circ S - S \circ T_W) g}\\
        &\le \sup_{g \colon [0,1] \to [0,1]} \normRTT{T_U \circ S - S \circ T_W} \normT{g}\\
        &\le \normRTT{T_U \circ S - S \circ T_W}
    \end{align*}
    for every operator $S \colon \LT \to \LT$.
    Hence, the statement follows from \Cref{th:countingLemmaPathsGraphons}.
\end{proof}

With this explicit counting lemma, we obtain that two graphons
have distance zero in the path distance if and only if their
path homomorphism densities are the same.

\begin{lemma}
    \label{le:pathDistTZero}
    Let $U, W \in \graphons$ be graphons.
    Then, $\pathDistROp(U, W) = 0$ if and only if $t(P,U) = t(P,W)$ for every path $P$.
\end{lemma}
\begin{proof}
    If $\pathDistROp(U, W) = 0$, then \Cref{th:countingLemmaPathDistT}
    yields that $t(P, U) = t(P, W)$ for every path $P$.
    On the other hand, if $t(P, U) = t(P, W)$ for every path $P$,
    then there is a signed Markov operator $S \in \signedMarkovs$ with
    $T_U \circ S = S \circ T_W$ by \Cref{th:pathHomsGraphons}.
    Then, $\pathDistROp(U, W) = 0$ follows immediately from the definition.
\end{proof}

By definition, the path distance is bounded from above by the tree distance
(with the appropriate norm), which means that it also is compatible
with the cut distance.

\begin{lemma}
    \label{le:pathDistTLeNeiDistT}
    \label{le:pathDistTCutDistConvergence}
    Let $U, W \in \graphons$ be graphons.
    Then, $\pathDistROp(U, W) \le \neiDistROp(U, W)$.
\end{lemma}

With these lemmas, we can apply \Cref{th:invCountingTrees}
and obtain the following inverse counting lemma for the path distance.

\begin{corollary}[Inverse Counting Lemma for $\pathDistROp$]
    \label{co:invCountingPathT}
    For every $\epsilon > 0$, there are $k > 0$ and $\eta > 0$ such that,
    for all graphons $U, W \in \graphons$,
    if
    $\lvert t(P, U) - t(P, W) \rvert \le \eta$
    for every path $P$ on at most $k$ vertices, then
    $\pathDistROp(U, W) \le \epsilon$.
\end{corollary}

\section{The Color Distance}
\label{subsec:cutDistanceInvariant}

\textit{Color Refinement}, also known as the \textit{1-dimensional Weisfeiler-Leman algorithm}, is a heuristic graph isomorphism test.
It computes a coloring of the vertices of a graph in a sequence of \textit{refinement rounds};
we say that color refinement \textit{distinguishes} two graphs if the computed color patterns differ.
Formally, for a graph $G$, we let $C^G_0(u) = 1$ for every $u \in V(G)$ and
$C^G_{i+1}(u) = \{\!\!\{C^G_i(v) \mid uv \in E(G)\}\!\!\}$
for every $i \ge 0$.
Let $C^G_{\infty} = C^G_i$ for the smallest $i$ such that $C^G_i(u) = C^G_i(v) \iff C^G_{i+1}(u) = C^G_{i+1}(v)$ for all $u,v \in G$ (\enquote{$C_i$ is stable}).
Then, color refinement distinguishes two graphs $G$ and $H$ if there is an $i \ge 0$
such that $\{\!\!\{ C^G_i(v) \mid v \in V(G) \}\!\!\} \neq \{\!\!\{ C^H_i(v) \mid v \in V(H) \}\!\!\}$.
It is well-known that the partition $\{ C^{-1}_\infty(i) \mid i \in C_{\infty}(V(G))\}$
is the \textit{coarsest equitable partition} of $V(G)$,
where a partition $\Pi$ of $V(G)$ is called \textit{equitable}
if for all $P, Q \in \Pi$ and $u, v \in P$,
the vertices $u$ and $v$ have the same number of neighbors in $Q$.

For a graph $G$, we can define a weighted graph $\ColG$
by letting
$V(\ColG) \coloneqq \{ C^{-1}_\infty(i) \mid i \in C_{\infty}(V(G))\}$,
$\alpha_C(\ColG) \coloneqq \lvert C \rvert$ for $C \in V(\ColG)$,
and $\beta_{CD}(\ColG) \coloneqq {M^G_{CD}}/{\lvert D \rvert}$
for all $C, D \in V(\ColG)$,
where $M^G_{CD}$ is the number of neighbors a vertex from $C$ has in $D$,
which is the same for all vertices in $C$ as
the partition induced by the colors of $C_{\infty}^G$ is equitable.
Note that we have
$\lvert C \rvert {M^G_{CD}} = \lvert D \rvert {M^G_{DC}}$
as both products describe the number of edges between $C$ and $D$, i.e.,
$\ColG$ is well-defined.
Usually, when talking about the invariant $\mathcal{I}^2_{C}$ computed by color refinement (see, e.g., \cite{KieferSchweitzerSelman15}), one does not normalize $M^G_{CD}$ by $\lvert D \rvert$.
However, by doing so, we do not only get a weighted graph (with symmetric
edge weights), but the graphs $G$ and $\ColG$ actually
have the same tree homomorphism counts.
Grebík and Rocha
already introduced the graphon analogue $U/\mathcal{C}(U)$ of $\ColG$
and proved the same fact for it \cite[Corollary $4.3$]{GrebikRocha2019};
hence, we omit the proof.

\begin{lemma}
    \label{le:colorGraphLemma}
    Let $T$ be a tree, and let $G$ be a graph.
    Then,
    $\hom(T, G) = \hom(T, \ColG)$.
\end{lemma}

By the result of Dvořák \cite{Dvorak2010},
$\ColG$ and $\Col{H}$ are isomorphic if and only if
$G$ and $H$ have the same tree homomorphism counts.
Hence, it is tempting to define a tree distance-like similarity measure
on graphs by simply considering the cut distance of $\ColG$ and $\Col{H}$.
For graphs $G$ and $H$, we call $\colDist(G, H) \coloneqq \dist(\ColG, \Col{H})$
the \textit{color distance} of $G$ and $H$.
As the cut distance $\dist$ is a pseudometric on graphs,
so is $\colDist$.
For $\colDist$, we immediately obtain a quantitative counting lemma
from \Cref{th:countingLemmaGraphons}
and \Cref{le:colorGraphLemma}.

\begin{corollary}[Counting Lemma for $\colDist$]
    Let $T$ be a tree, and let $G$ and $H$ be graphs.
    Then,
    $\lvert \dens(T, G) - \dens(T, H) \rvert \le \lvert E(T) \rvert \cdot \colDist(G, H)$.
\end{corollary}

Clearly, $\neiDist$ and $\colDist$ have the same graphs of distance zero.
Moreover, one can easily verify that the tree distance is bounded
from above by the color distance.

\begin{lemma}
    \label{le:neiDistLeColDist}
    Let $G$ and $H$ be graphs.
    Then, $\neiDist(G, H) \le \colDist(G, H)$.
\end{lemma}
\begin{proof}
    We have
    $\neiDist(G, H) = \neiDist(\ColG, \Col{H}) \le \dist(\ColG, \Col{H}) = \colDist(G, H)$
    by \Cref{le:colorGraphLemma} and \Cref{le:neiDistLeCutDist}.
\end{proof}

Now, the obvious question is whether these pseudometrics are the same
or, at least, define the same topology.
But it is not hard to find a counterexample;
the color distance sees differences between graphs
that the tree distance and tree homomorphisms do not see.
In particular, an inverse counting lemma cannot hold for the
color distance.
See \Cref{fig:le:notEquivalent}, and for the moment,
assume that we can construct a sequence $(G_n)_n$ of graphs such
that $\Col{G_n}$ is as depicted.
It is easy to verify that $\dist(\Col{G_n}, K_3) \xrightarrow{n \rightarrow \infty} 0$,
and thus, both $\neiDist(G_n, K_3) \xrightarrow{n \rightarrow \infty} 0$
and $\lvert t(T, G_n) - t(T, K_3) \rvert \xrightarrow{n \rightarrow \infty} 0$ for every tree $T$.
But, $\colDist(G_n, K_3) \ge \frac{1}{3} \cdot \frac{1}{3} \cdot \frac{2}{3}$ for every $n$
since $\Col{G_n}$ has a vertex without a loop.

\begin{figure}
    \begin{center}
\begin{tikzpicture}
		\node[vertex, label={90:$n$}] (E1) {};
        \node[vertex, label={270:$n$}, below left = 1cm and 0.55cm of E1] (E2) {};
		\node[vertex, label={270:$n$}, below right = 1cm and 0.55cm of E1] (E3) {};
        \node[above = 0.6cm of E1] (label) {$\Col{G_n}$};

        \path[draw, thick] (E1) edge node[xshift = -6pt, yshift = 4pt] {$\frac{n}{n}$} (E2);
        \path[draw, thick] (E1) edge node[xshift = 8pt, yshift = 4pt] {$\frac{n-1}{n}$} (E3);
        \path[draw, thick] (E2) edge node[xshift = 0pt, yshift = -8pt] {$\frac{n-2}{n}$} (E3);

		\node[vertex, left = 5cm of E1] (v1) {};
        \node[vertex, below left = 1cm and 0.55cm of v1] (v2) {};
		\node[vertex, below right = 1cm and 0.55cm of v1] (v3) {};
        \node[above = 0.6cm of v1] (label2) {$K_3$};

        \path[draw, thick] (v1) edge (v2);
        \path[draw, thick] (v1) edge (v3);
        \path[draw, thick] (v2) edge (v3);

        \path ($(E1)!0.5!(E2)$) edge[->, draw, thick, shorten <= 1.0cm, shorten >= 1.0cm] node[yshift= 8pt] {$\dist$, $\neiDist$} ($(v1)!0.5!(v3)$);

        \node[right = 5cm of E1] (w1ghost) {};
        \node[below left = 1cm and 0.55cm of w1ghost] (w2ghost) {};
        \node[vertex, below = 0.275 of w1ghost, yshift = -3pt, label={270:$3$}] (w1) {};
        \node[above = 0.6cm of w1ghost] (label3) {$\Col{K_3}$};

        \path[every loop/.style={in = 120, out = 60}, draw, thick] (w1) edge[loop above] node[] {$\frac{2}{3}$} (w1);

        \path ($(E1)!0.5!(E3)$) edge[->, draw, thick, shorten <= 1.0cm, shorten >= 1.0cm] node[yshift= 8pt] {\cancel{$\dist$}, $\neiDist$} ($(w1ghost)!0.5!(w2ghost)$);
    \end{tikzpicture}
    \vspace{-0.6cm}
    \end{center}
    \caption{An example separating the color distance from the tree distance}
    \label{fig:le:notEquivalent}
\end{figure}
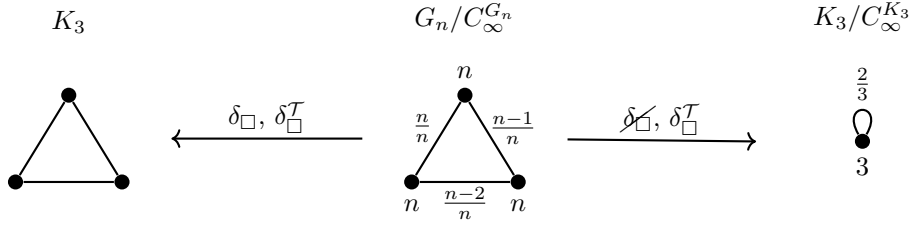

The existence of graphs $G_n$
such that $\Col{G_n}$ is as depicted in \Cref{fig:le:notEquivalent}
follows easily from inversion results
for the color refinement invariant $\mathcal{I}^2_{C}$.
Otto first proved that $\mathcal{I}^2_{C}$
admits polynomial time inversion on structures \cite{Otto97},
and Kiefer, Schweitzer, and Selman
gave a simple construction to show that $\mathcal{I}^2_{C}$
admits linear-time inversion on the class of graphs \cite{KieferSchweitzerSelman15}.
Basically, we partition $3n$ vertices
into three sets of size $n$
and add edges between these partitions
such that they induce $n$-, $(n-1)$-, and $(n-2)$-regular bipartite graphs.

The example in \Cref{fig:le:notEquivalent} actually answers
an open question of Grebík and Rocha \cite[Question $3.1$]{GrebikRocha2019}.
They ask whether the set $\{W/\mathcal{C}(W) \mid W \in \graphonsWI \}$
is closed in $\graphonsWI$:
it is not.
With a more refined argument, we can actually show
that $\{W_{\ColG} \mid G \text{ graph}\}$ is already dense in $\graphonsWI$.
By properly rounding the weights of a given weighted graph,
we can turn the inversion result of \cite{KieferSchweitzerSelman15}
into a statement about approximate inversion.
The proof of \Cref{th:colRefApproxInv} can be found in \Cref{subsec:inverseApproximationProof}

\begin{theorem}
    \label{th:colRefApproxInv}
    Let $H$ be a weighted graph.
    For every $n \ge 2 \cdot \vs(H)$, there is a graph $G$ on $n^2$ vertices such that
    $\dist(\ColG, H) \le 3 \cdot {\vs(H)}/{n} + \frac{1}{4} \cdot ({\vs(H)}/{n})^2$.
\end{theorem}

In \Cref{th:colRefApproxInv}, the size of the resulting graph
depends on how close we want it to be to the input graph.
A simple consequence of the compactness of the graphon space
is that, for $\epsilon > 0$, we can approximate any graphon
with an error of $\epsilon$ in $\dist$
by
a graph on $N(\epsilon)$ vertices, where $N(\epsilon)$
is independent of the graphon \cite[Corollary $9.25$]{Lovasz2012}.
With \Cref{th:colRefApproxInv}, this implies that the same is possible
with the weighted graphs $\ColG$.
This also means that the closure of the set $\{W_{\ColG} \mid G \text{ graph}\}$
is already $\graphonsWI$.

\section{Conclusions}
\label{sec:conclusion}

We have introduced similarity measures for graphs
that can be formulated as convex optimization problems
and shown surprising
correspondences to tree and path homomorphism densities.
This takes previous results on the \enquote{expressiveness} of
homomorphism counts from an exact to an approximate level.
Moreover, it helps to give a theoretical understanding of kernel methods
in machine learning,
which are often based on counting certain substructures in graphs.
Proving the correspondences to homomorphism densities
was made possible by introducing
our similarity measures
for the more general case of graphons,
where tools from functional analysis let us prove
the general statement that
every \enquote{reasonably defined} pseudometric has
to satisfy a correspondence to homomorphism densities.

Various open questions remain.
The compactness argument used in \Cref{sec:pseudometrics}
only yields non-quantitative statements.
Hence, we do not know how
close the graphs have to be in the pseudometric for their
homomorphism densities to be close and vice versa.
Only for paths we were able to prove
a quantitative counting lemma,
which uses the same factor $\es(F)$ as the counting lemma for general graphs.
It seems conceivable that a
quantitative counting lemma for trees that uses the same factor $\es(T)$
also holds.
As the proof of the quantitative inverse counting lemma is quite involved \cite{BorgsEtAl2008, Lovasz2012},
proving such statements for trees and paths
should not be easy.

More in reach seems to be the question of how the tree distance
generalizes to the class $\mathcal{T}_k$ of graphs of treewidth at most $k$.
Homomorphism counts from graphs in $\mathcal{T}_k$
can also be characterized
in terms of linear equations in the case of graphs \cite{Dvorak2010} (see also \cite{Dell2018}).
How does such a characterization for graphons look like?
And how does one define a distance measure from this?

Another open question concerns further characterizations of
fractional isomorphism, e.g., the color refinement algorithm,
which gives a characterization based
on equitable partitions.
Can one prove a correspondence between the tree distance
and, say, $\epsilon$-equitable partitions?
It is not hard to come up with a definition for such partitions;
the hard part is to prove that graphs that are similar in
the tree distance possess such a partition.

\newpage
\appendix
\section{Appendix}
\label{sec:appendix}
\label{sec:facts}

This appendix contains the proofs omitted from the main body of the paper
and collects some results used in these proofs.
We start with some additions to the preliminaries.

For $1 \le p < \infty$, let $\Lp$ denote the space of
$\R$-valued $p$-integrable functions on $[0,1]$ (modulo equality almost anywhere).
Likewise, $\LI$ denotes the space of essentially bounded $\R$-valued functions on
$[0,1]$ (modulo equality almost anywhere).
Unless explicitly stated otherwise, the functions that we consider are $\R$-valued.
Let $1 \le p < q \le \infty$.
By Hölder's inequality, we have $\lVert f \rVert_p \le \lVert f \rVert_q$
for every function $f \in \Lq$ since $[0,1]$ has measure one.
In particular, we have $\Lq \subseteq \Lp$.
Among these spaces, $\LT$ plays a special role as it is a Hilbert space
as mentioned in the preliminaries.
For an operator $T \colon \Lp \to \Lq$, where $1 \le p,q \le \infty$,
let $\normRpq{T}$ denote its operator norm, i.e.,
$\normRpq{T} = \sup_{\lVert g \rVert_p \le 1} \lVert T g \rVert_q$,
and let $\normpq{T}$ be the operator norm when viewing $T$ as an operator
on the complex $\Lp$, i.e.,
$\normpq{T} = \sup_{\lVert g \rVert_p \le 1, g \colon [0,1] \to \C} \lVert T g \rVert_q$.

For $p \in [1, \infty]$, an operator $S \colon \Lp \to \Lp$ is
called a Markov operator if
$S \ge 0$ (\enquote{$S$ is positive}), i.e., $f \ge 0$ implies $S(f) \ge 0$,
$S(\allOne) = \allOne$, and
$\int_{[0,1]} (S f) (x) \dx = \int_{[0,1]} f(x) \dx$ for every $f \in \Lp$.
Here, $\allOne$ is the all-one function on $[0,1]$.
For $\LT$, the third condition can be reformulated as $S^*(\allOne) = \allOne$,
where $S^*$ is the Hilbert adjoint of $S$.
Unless explicitly stated otherwise, we work with Markov operators
$S \colon \LT \to \LT$ and denote the set of all such operators by $\markov$.
By \Cref{th:markovRestriction}, it does not really matter
which space $\Lp$ one considers Markov operators on.
Also note that the results on Markov operators
in \Cref{sec:markov} are originally stated for complex $\Lp$ spaces.
Since we work with graphons, which are $\R$-valued, and Markov operators
map $\R$-valued functions to $\R$-valued functions, cf.\ \Cref{le:posOperators},
this does not make a different for us.

Recall that every kernel $W \in \kernels$ defines an
operator $T_W$ by setting
$(T_W f)(x) = \int_{[0,1]} W(x,y) f(y) \dy$
for every $x \in [0,1]$.
Unless specified otherwise, we view it as an operator $T_W \colon \LT \to \LT$.
Then it is a Hilbert-Schmidt operator, and in particular, compact \cite{Lovasz2012}.
The definition of $T_W$ also allows to view it as an operator
$T_W \colon \colon \LO \to \LI$, and hence,
by the aforementioned inclusions of $L_p$ spaces, we can view $T_W$
as an operator $T_W: L_p[0,1] \to L_q[0,1]$ for all $1 \le p,q \le \infty$.

\subsection{Operators}
\label{sec:fanalysis}

\begin{lemma}[{\cite[Theorem $12.7$]{Rudin1991}}]
    \label{le:operatorIsZero}
    Let $T$ be a bounded linear operator on a Hilbert space $\mathcal{H} \neq \{0\}$.
    If $\langle T g, g \rangle = 0$ for every $g \in \mathcal{H}$, then $T = 0$.
\end{lemma}

\begin{theorem}[{Riesz-Thorin Interpolation Theorem, e.g., \cite[Theorem $1.1.1$]{Bergh1976}}]
    \label{th:rieszThorin}
    In the following, all $\Lp$ spaces are complex.
    Assume that $p_0 \neq p_1$, $q_0 \neq q_1$,
    $T \colon L_{p_0}(X, \mathcal{S}, \mu) \to L_{q_0}(Y, \mathcal{T}, \nu)$
    with norm $\lVert T \rVert_{p_0 \rightarrow q_0}$, and
    $T \colon L_{p_1}(X, \mathcal{S}, \mu) \to L_{q_1}(Y, \mathcal{T}, \nu)$
    with norm $\lVert T \rVert_{p_1 \rightarrow q_1}$.
    Then,
    $T \colon L_{p}(X, \mathcal{S}, \mu) \to L_{q}(Y, \mathcal{T}, \nu)$
    with norm
    $\lVert T \rVert_{p \rightarrow q} \le \lVert T \rVert_{p_0 \rightarrow q_0}^{1 - \theta} \lVert T \rVert_{p_1 \rightarrow q_1}^{\theta}$
    provided that $0 < \theta < 1$ and
    $\frac{1}{p} = \frac{1 - \theta}{p_0} + \frac{\theta}{p_1}$,
    $\frac{1}{q} = \frac{1 - \theta}{q_0} + \frac{\theta}{q_1}$.
\end{theorem}

\subsection{Graphons}
\label{sec:graphonfacts}

\begin{lemma}[{\cite[Equation ($7.2$)]{Lovasz2012}}]
    \label{le:steppingFunctionDensity}
    Let $F$ be a graph and $H$ be a weighted graph.
    Then,
    $t(F, H) = t(F, W_H)$.
\end{lemma}

\begin{lemma}[{\cite[Lemma $8.9$]{Lovasz2012}}]
    \label{le:cutDistWeightedGraphsGraphons}
    Let $G$ and $H$ be weighted graphs.
    Then,
    $\dist(G, H) = \dist(W_G, W_H)$.
\end{lemma}

\begin{lemma}[{\cite[Lemma $8.10$]{Lovasz2012}}]
    \label{le:cutNormAttained}
    For any kernel $W \in \kernels$, the optima
    \begin{align*}
        &\sup_{S, T \subseteq [0,1]} \left\lvert \int_{S \times T} W(x,y) \dx \dy \right\rvert&&\text{and}&&\sup_{f, g \colon [0,1] \to [0,1]} \left\lvert \int_{[0,1]^2} f(x) g(y) W(x,y) \dx \dy \right\rvert&
    \end{align*}
    are attained and they are both equal to $\lVert W \rVert_\square$.
\end{lemma}

Let
\begin{equation*}
    \normCT{W} \coloneqq \sup_{f,g \colon [0,1] \to [0,1]} \left\lvert \int_{[0,1]^2} W(x,y) f(x) g(y) \dx \dy \right\rvert
\end{equation*}
for a kernel $W \in \kernels$.
Then, $\normC{W} \le \normCT{W} \le 4 \normC{W}$
for every kernel $W \in \kernels$ \cite[$(4.4)$]{Janson2013}.
The following lemma is a consequence of the Riesz-Thorin
interpolation theorem, \Cref{th:rieszThorin}.
\begin{lemma}[{\cite[Lemma E.$6$]{Janson2013}}]
    \label{le:normsForGraphons}
    If $\lvert W \rvert \le 1$, then for all $p,q \in [1, \infty]$,
    \begin{equation*}
        \normCT{W} = \normRIO{T_W} \le \normRpq{T_W} \le \sqrt{2} \normCT{W}^{\min(1-1/p, 1/q)}.
    \end{equation*}
    Consequently, for any fixed $p > 1$ and $q < \infty$,
    if $W_1, W_2, \dots$ and $W$ are graphons (defined on the same space),
    then $\normC{W_n - W} \xrightarrow{n \rightarrow \infty} 0$
    if and only if
    $\normRpq{T_{W_n} - T_W} \xrightarrow{n \rightarrow \infty} 0$.
\end{lemma}

\subsection{Markov Operators}
\label{sec:markov}

Here, we collect some facts on Markov operators from \cite{EisnerEtAl2015},
where they are stated for complex $\Lp$ spaces.
However, by \Cref{le:posOperators}, we can also consider real $\Lp$ spaces
instead.

\begin{lemma}[{\cite[Example $13.1$, 3)]{EisnerEtAl2015}}]
    \label{th:koopmanMarkov}
    Every Koopman operator $T_\varphi$ associated with a
    measure-preserving system $(X;\varphi)$ is a Markov operator.
\end{lemma}

\begin{lemma}[{\cite[Lemma $7.5$ b), c)]{EisnerEtAl2015}}]
    \label{le:posOperators}
    Let $E, F$ be Banach lattices, and let $S \colon E \to F$ be a positive
    operator.
    Then, the following assertions hold:
    \begin{enumerate}
        \item $S f \in F_\R$ for every $f \in E_\R$.
        \item $S(\Re f) = \Re S f$ and $S(\Im f) = \Im S f$ for every $f \in E$.
    \end{enumerate}
\end{lemma}

\begin{lemma}[{\cite[Theorem $13.2$ a), c)]{EisnerEtAl2015}}]
    \label{le:markovClosure}
    The set of Markov operators is closed under composition and adjoints.
\end{lemma}

\begin{lemma}[{\cite[Theorem $13.2$ b)]{EisnerEtAl2015}}]
    \label{le:markovContraction}
    Every Markov operator $S \colon \LO \to \LO$
    is a Dunford-Schwartz operator, i.e.,
    it restricts to a contraction on each space $L_p[0,1]$
    for $1 \le p \le \infty$, i.e.,
$\lVert S f \rVert_p \le \lVert f \rVert_p$
for every $f \in \Lp$.
\end{lemma}

\begin{theorem}[{\cite[Proposition $13.6$]{EisnerEtAl2015}}]
    \label{th:markovRestriction}
    Let $X, Y$ be probability spaces, and let $1 \le p \le \infty$.
    Then, the restriction mapping
    \begin{align*}
        &\Phi_p \colon M(X; Y) \to M_p(X; Y),& &\Phi_p(S) \coloneqq \restr{S}{L_p}
    \end{align*}
    is a bijection satisfying $\Phi_p(S') = \Phi_p(S)'$ for every $S \in M(X; Y)$.
    Finally, for $1 \le p < \infty$, the mapping $\Phi_p$ is a homeomorphism
    for the weak as well as the strong operators topologies.
\end{theorem}

By \Cref{th:markovRestriction}, the weak operator topology on the set of
Markov operators does not change when one considers Markov operators
as mappings $\Lp \to \Lp$ for different $p \in [1, \infty)$.

\begin{theorem}[{\cite[Theorem $13.8$]{EisnerEtAl2015}}]
    \label{th:markovCompact}
    The set of Markov operators is compact with respect
    to the weak operator topology.
\end{theorem}

\subsection{Proof of \Cref{le:steppingFunctionNeiDist} and \Cref{le:steppingFunctionPathDistSpec} (Definitions Coincide)}
\label{subsec:distancesCoincide}
We first prove \Cref{le:steppingFunctionNeiDist}, i.e.,
that the definitions of the tree distance for graphs coincide with
the ones for graphons.
In this subsection, we fix two graphs $G$ and $H$,
where we w.l.o.g.\ assume that $V(G) = \{1, \dots, n\}$
and $V(H) = \{1, \dots, m\}$.
Let $I_1, \dots, I_n$ and $J_1, \dots, J_m$ be the partitions of $[0,1]$
into the steps of $W_G$ and $W_H$, respectively,
such that $I_i$ and $J_j$ correspond to vertex $i \in V(G)$ and $j \in V(H)$,
respectively.
Let $A \in \R^{n \times n}$ and $B \in \R^{m \times m}$ be the
adjacency matrices of $G$ and $H$, respectively.
For a matrix $X \in \R^{n \times m}$,
let the kernel $W_X$ be given by setting $W_X(x,y) \coloneqq nm \cdot X_{ij}$
for all $x \in I_i$, $y \in J_j$, $i \in \numTo{n}$, $j \in \numTo{m}$.
Then, let $T_X \coloneqq T_{W_X}$ be the operator defined by $W_X$, i.e.,
\begin{equation*}
    (T_X f) (x) = \int_{[0,1]} W_X(x, y) f(y) \dy
\end{equation*}
for every $x \in [0,1]$.

\begin{lemma}
    \label{le:fracOverlayToMarkov}
    For a fractional overlay $X \in \R^{n \times m}$,
    the operator $T_X$ is a Markov operator with Hilbert adjoint $T_X^* = T_{X^T}$.
\end{lemma}
\begin{proof}
    For $x \in I_i$, we have
    \begin{equation*}
        (T_X \allOne) (x)= \int_{[0,1]} W_X(x, y) \dy = \sum_{j \in \numTo{m}} \int_{J_j} nm \cdot X_{ij} \dy = \sum_{j \in \numTo{m}} n \cdot X_{ij}
        = \allOne.
    \end{equation*}
    For $f,g \in \LT$, we have
    \begin{align*}
        \langle T_X f, g \rangle = \int_{[0,1]} \int_{[0,1]} W_X(x,y) f(y) \dy \, {g(x)} \dx
&= \int_{[0,1]} \int_{[0,1]} W_{X^T}(y,x) f(y) \dy \, {g(x)} \dx\\
&= \int_{[0,1]} f(y) {\int_{[0,1]} W_{X^T}(y,x) g(x) \dx} \dy\\
&= \langle f, T_{X^T} g \rangle,
    \end{align*}
    where we, of course, used the Theorem of Fubini.
    Hence, $T_{X^T}$ is the Hilbert adjoint of $T_X$,
    and since $X$ is arbitrary, we also have $T_{X^T}(\allOne) = \allOne$.
    Therefore, $T_X$ is a Markov operator.
\end{proof}

Let $S \colon \LT \to \LT$ be an operator.
For $i \in \numTo{n}$, $j \in \numTo{m}$, we define
\begin{equation*}
    (X_S)_{ij} \coloneqq \int_{I_i} S(\cFun_{J_j}) (x) \dx.
\end{equation*}

\begin{lemma}
    \label{le:markovToFracOverlay}
    For a Markov operator $S \colon \LT \to \LT$,
    the matrix $X_S$ is a fractional overlay of $G$ and $H$.
\end{lemma}
\begin{proof}
    We have $X_S \in \R^{n \times m}$
    with non-negative entries.
    For $i \in \numTo{n}$, the linearity of $S$ yields
    \begin{align*}
        \sum_{j \in \numTo{m}} (X_S)_{ij}= \sum_{j \in \numTo{m}} \int_{I_i} S(\cFun_{J_j}) (x) \dx = \int_{I_i} S(\allOne) (x) \dx = \int_{I_i} 1 \dx = \frac{1}{n}.
    \end{align*}
    For $j \in \numTo{m}$, we get
    \begin{align*}
        \sum_{i \in \numTo{n}} (X_S)_{ij}= \sum_{i \in \numTo{n}} \int_{I_i} S(\cFun_{J_j}) (x) \dx = \int_{[0,1]} S(\cFun_{J_j}) (x) \dx = \int_{[0,1]} \cFun_{J_j} (x) \dx = \frac{1}{m}.
    \end{align*}
\end{proof}

\begin{proof}[Proof of \Cref{le:steppingFunctionNeiDist}, First Equality]
    First, we prove that $\neiDist(G, H) \ge \neiDist(W_G, W_H)$.
    Let $X \in \R^{n \times m}$ be a fractional overlay of $G$ and $H$.
    By \Cref{le:fracOverlayToMarkov},
    the operator $T_X$ is a Markov operator.
    For a measurable function $f \colon [0,1] \to [0,1]$, let
    $u^f \in \R^n$ be the vector with $u^f_i = n \cdot \int_{I_i} f(x) \dx$ for every $i \in \numTo{i}$,
    and let $v^f \in \R^m$ be the vector
    with $v^f_j = m \cdot \int_{J_j} f(x) \dx$ for every $j \in \numTo{m}$.
    For $x \in I_i$ and a measurable function $g \colon [0,1] \to [0,1]$, we have
    \begin{align*}
        ((T_{W_G} \circ T_X) g) (x)&= \int_{[0,1]} W_G(x,y) \left(\int_{[0,1]} W_X(y,z) g(z) \dz \right) \dy\\
&= \sum_{k \in \numTo{n}} \int_{I_k}  A_{ik} \left(\sum_{j \in \numTo{m}} \int_{J_j} nm \cdot X_{kj} g(z) \dz \right) \dy\\
        &= \sum_{k \in \numTo{n}} A_{ik} \int_{I_k}  \left(\sum_{j \in \numTo{m}} n \cdot X_{kj} v^g_j\right) \dy\\
&= (A X v^g)_{i}.
    \end{align*}
    Hence, we get
    \begin{align*}
        \langle f, (T_{W_G} \circ T_X) g \rangle &= \int_{[0,1]} f(x) ((T_{W_G} \circ T_X) g) (x) \dx\\
&= \sum_{i \in \numTo{n}} \int_{I_i} f(x) (A X v^g)_i \dx\\
&= \frac{1}{n} \cdot \sum_{i \in \numTo{n}} (A X v^g)_i u^f_i\\
        &= \frac{1}{nm} \cdot {u^f}^T (m \cdot A X) v^g
    \end{align*}
    for all measurable $f,g \colon [0,1] \to [0,1]$.
    In a similar fashion, one can verify that,
    for $x \in I_i$ and a measurable function $g \colon [0,1] \to [0,1]$, we have
    \begin{equation*}
        ((T_X \circ T_{W_H}) g) (x)= (\frac{n}{m} \cdot X B v^g)_{i}
    \end{equation*}
    and, thus,
    \begin{equation*}
        \langle f, (T_{W_G} \circ T_X) g \rangle = \frac{1}{nm} \cdot {u^f}^T (n \cdot X B) v^g
    \end{equation*}
    for all measurable $f,g \colon [0,1] \to [0,1]$.
    Combining this yields
    \begin{align*}
        \neiDist(W_G, W_H)&= \inf_{S \in \markov} \sup_{f, g \colon [0,1] \to [0,1]} \lvert \langle f, (T_{W_G} \circ S - S \circ T_{W_H}) g \rangle \rvert\\
        &\le \inf_{X \in \mathcal{X}(G, H)} \sup_{f, g \colon [0,1] \to [0,1]} \lvert \langle f, (T_{W_G} \circ T_X - T_X \circ T_{W_H}) g \rangle \rvert\\
        &= \inf_{X \in \mathcal{X}(G, H)} \frac{1}{nm} \cdot \sup_{f, g \colon [0,1] \to [0,1]} \lvert {u^f}^T (m \cdot AX - n \cdot XB) v^g \rvert\\
        &= \inf_{X \in \mathcal{X}(G, H)} \frac{1}{nm} \normC{m \cdot AX - n \cdot XB}\\
        &= \neiDist(G, H)
    \end{align*}
    since the maximum is attained at $0$-$1$-vectors.

    To prove that $\neiDist(G, H) \le \neiDist(W_G, W_H)$,
    let $S \colon \LT \to \LT$ be a Markov operator.
    By \Cref{le:markovToFracOverlay}, $X_S \in \R^{n \times m}$
    is a fractional overlay of $G$ and $H$.
    For a set $P \subseteq \numTo{n}$, let
    $\cFun^I_{P} \coloneqq \sum_{i \in P} \cFun_{I_i}$
    and, for a set $Q \subseteq \numTo{m}$,
    let $\cFun^J_{Q} \coloneqq \sum_{j \in Q} \cFun_{J_j}$.
    Then, for $x \in I_i$, the linearity of $S$ yields
    \begin{align*}
        ((T_{W_G} \circ S) \cFun^J_{Q}) (x)= \int_{[0,1]} W_G(x,y) (S \cFun^J_Q) (y) \dy
        &= \sum_{j \in Q} \int_{[0,1]} W_G(x,y) (S \cFun_{J_j}) (y) \dy\\
        &= \sum_{j \in Q} \sum_{k \in \numTo{n}} \int_{I_k} A_{ik} (S \cFun_{J_j}) (y) \dy\\
        &= \sum_{j \in Q} \sum_{k \in \numTo{n}} A_{ik} (X_S)_{kj}\\
        &= \sum_{j \in Q} (AX_S)_{ij}
    \end{align*}
    and, thus,
    \begin{align*}
        \langle \cFun^I_{P}, (T_{W_G} \circ S) \cFun^J_{Q} \rangle = \int_{[0,1]} \cFun^I_{P} (x) ((T_{W_G} \circ S) \cFun^J_{Q})(x) \dx
        &= \sum_{i \in \numTo{n}} \int_{I_i} \cFun^I_{P} (x) \big(\sum_{j \in Q} (AX_S)_{ij}\big) \dx\\
        &= \sum_{\substack{i \in \numTo{n},\\ j \in Q}} (AX_S)_{ij} \int_{I_i} \cFun^I_{P} (x) \dx\\
        &= \frac{1}{nm} \sum_{\substack{i \in P,\\ j \in Q}} (m \cdot AX_S)_{ij}.
    \end{align*}
    In a similar fashion, one can show that, for a set $Q \subseteq \numTo{m}$
    and $x \in [0,1]$,
    we have
    \begin{align*}
        ((S \circ T_{W_H}) \cFun^J_{Q}) (x)= \sum_{k \in \numTo{m}} \big(\sum_{j \in T} \frac{1}{m} B_{kj}\big) S(\cFun_{J_k}) (x)
    \end{align*}
    and, thus,
    \begin{align*}
        \langle \cFun^I_{P}, (S \circ T_{W_H}) \cFun^J_{Q} \rangle = \frac{1}{nm} \sum_{\substack{i \in P,\\ j \in Q}} (n \cdot X_S B)_{ij}.
    \end{align*}
    Combining this yields
    \begin{align*}
        \neiDist(G, H)
        &= \inf_{X \in \mathcal{X}(G, H)} \frac{1}{nm} \normC{m \cdot AX - n \cdot XB}\\
        &\le \inf_{S \in \markov} \frac{1}{nm} \normC{m \cdot A X_S - n \cdot X_S B}\\
        &= \inf_{S \in \markov} \sup_{\substack{P \subseteq \numTo{n}, Q \subseteq \numTo{m}}} \lvert \langle \cFun^I_P, (T_{W_G} \circ S - S \circ T_{W_H}) \cFun^J_Q \rangle \rvert\\
        &\le \neiDist(W_G, W_H).
    \end{align*}
\end{proof}

Proving the second equality is similar, although a bit more complicated
due to the non-linearity of the square function.
Here, we have to apply the
the Cauchy-Schwarz inequality
at certain points.
\begin{proof}[Proof of \Cref{le:steppingFunctionNeiDist}, Second Equality]
    First, we prove that $\neiDistSpec(G, H) \ge \neiDistROp(W_G, W_H)$.
    Let $X \in \R^{n \times m}$ be a fractional overlay of $G$ and $H$.
    By \Cref{le:fracOverlayToMarkov},
    the operator $T_X$ is a Markov operator.
    For $g \colon [0,1] \to \R$ with $\normT{g} \le 1$,
    let $v^g \in \R^m$ be given by $v^g_j \coloneqq \sqrt{m} \cdot \int_{J_j} g(x) \dx$ for every $j \in \numTo{m}$.
    Then,
    \begin{align*}
        \normT{v^g}^2
= \sum_{j \in [m]} m \cdot \big( \int_{J_j} g(x) \dx\big)^2
        &\le \sum_{j \in [m]} m \cdot \frac{1}{m} \cdot \int_{J_j} g(x)^2 \dx \tag*{(Cauchy-Schwarz)}\\
&= \int_{[0,1]} g(x)^2 \dx\\
        &= \normT{g}^2,
    \end{align*}
    that is, $\normT{v^g} \le \normT{g} \le 1$.
    For $x \in I_i$, we have
    \begin{align*}
        ((T_{W_G} \circ T_X) g) (x)&= \int_{[0,1]} W_G(x,y) \left(\int_{[0,1]} W_X(y,z) g(z) \dz \right) \dy\\
&= \sum_{k \in \numTo{n}} \int_{I_k} A_{ik} \left(\sum_{j \in \numTo{m}} \int_{J_j} nm \cdot X_{kj} g(z) \dz \right) \dy\\
        &= \sum_{k \in \numTo{n}} A_{ik} \int_{I_k}  \left(\sum_{j \in \numTo{m}} n \cdot \frac{m}{\sqrt{m}} \cdot X_{kj} v^g_j\right) \dy\\
&= \frac{1}{\sqrt{m}}(m \cdot A X v^g)_{i}.
    \end{align*}
    In a similar fashion, one can verify that,
    for $x \in I_i$, we have
    \begin{equation*}
        ((T_X \circ T_{W_H}) g) (x)= \frac{1}{\sqrt{m}}(n\cdot X B v^g)_{i}
    \end{equation*}
    and, thus,
    \begin{align*}
        \normT{(T_{W_G} \circ T_X - T_X \circ T_{W_H}) g}^2&= \int_{[0,1]} ((T_{W_G} \circ T_X - T_X \circ T_{W_H}) g(x))^2 \dx\\
        &= \sum_{k \in [n]} \int_{I_k} (\frac{1}{\sqrt{m}}(m \cdot A X - n \cdot X B) v^g)_{i}^2 \dx\\
        &= \sum_{k \in [n]} (\frac{1}{\sqrt{nm}}(m \cdot A X - n \cdot X B) v^g)_{i}^2\\
        &= \normT{\frac{1}{\sqrt{nm}}(m \cdot A X - n \cdot X B)v^g}^2.
    \end{align*}

    Hence,
    \begin{align*}
        \neiDistROp(W_G, W_H)&= \inf_{S \in \markov} \sup_{\normT{g} \le 1} \normT{(T_{W_G} \circ S - S \circ T_{W_H}) g}\\
        &\le \inf_{X \in \mathcal{X}(G, H)} \sup_{\normT{g} \le 1} \normT{(T_{W_G} \circ T_X - T_X \circ T_{W_H})g}\\
        &= \inf_{X \in \mathcal{X}(G, H)} \sup_{\normT{g} \le 1} \normT{\frac{1}{\sqrt{nm}}(m \cdot A X - n \cdot X B)v^g}\\
        &\le \inf_{X \in \mathcal{X}(G, H)} \sup_{\substack{v \in \R^m,\\ \normT{v} \le 1}} \normT{\frac{1}{\sqrt{nm}}(m \cdot A X - n \cdot X B)v}\\
        &= \neiDistSpec(G, H).
    \end{align*}

    To prove that $\neiDistSpec(G, H) \le \neiDistROp(W_G, W_H)$,
    let $S \colon \LT \to \LT$ be a Markov operator.
    By \Cref{le:markovToFracOverlay}, $X_S \in \R^{n \times m}$
    is a fractional overlay of $G$ and $H$.
    For $v \in \R^m$, let $g_v \coloneqq \sum_{j \in [m]} \sqrt{m} \cdot v_j \cFun_{J_j}$.
    Then,
    \begin{align*}
        \normT{g_v}^2= \int_{[0,1]} g_v(x)^2 \dx
        &= \sum_{j \in [m]} \int_{J_j} m \cdot v_j^2 \dx\\
        &= \sum_{j \in [m]} v_j^2\\
        &= \normT{v}^2,
    \end{align*}
    that is, $\normT{g_v} = \normT{v} \le 1$.
    Then, for $x \in I_i$, the linearity of $S$ yields
    \begin{align*}
        ((T_{W_G} \circ S) g_v) (x)&= \int_{[0,1]} W_G(x,y) (S g_v) (y) \dy\\
        &= \sqrt{m} \sum_{j \in [m]} v_j \int_{[0,1]} W_G(x,y) (S \cFun_{J_j}) (y) \dy\\
        &= \sqrt{m} \sum_{j \in [m]} v_j \sum_{k \in \numTo{n}} \int_{I_k} A_{ik} (S \cFun_{J_j}) (y) \dy\\
        &= \sqrt{m} \sum_{j \in [m]} v_j \sum_{k \in \numTo{n}} A_{ik} (X_S)_{kj}\\
        &= \frac{1}{\sqrt{m}} (m \cdot AX_Sv)_i.
    \end{align*}
    In a similar fashion, one can show that,
    for $x \in [0,1]$,
    we have
$((S \circ T_{W_H}) g_v) (x)= \frac{1}{\sqrt{m}} \sum_{k \in \numTo{m}} (Bv)_k S(\cFun_{J_k}) (x)$.
In the following,
    let $a_i \coloneqq \frac{1}{\sqrt{m}} (m \cdot AX_Sv)_i$
    and $b_i \coloneqq \frac{1}{\sqrt{n}} (n \cdot X_S B v)_i$
    for $i \in [n]$.
    Moreover,
    let $b(x) \coloneqq \frac{1}{\sqrt{m}} \sum_{k \in \numTo{m}} (Bv)_k S(\cFun_{J_k}) (x)$
    for $x \in [0,1]$.
    Then,
    \begin{align*}
        \normT{(T_{W_G} \circ S - S \circ T_{W_H}) g_v}^2
        &= \int_{[0,1]} \Big(((T_{W_G} \circ S - S \circ T_{W_H}) g_v) (x)\Big)^2 \dx\\
        &= \sum_{i \in [n]} \int_{I_i} \big(a_i - b(x)\big)^2 \dx\\
        &= \sum_{i \in [n]} \Bigg( \begin{aligned}[t]
               &\int_{I_i} a_i^2 \dx - \int_{I_i} 2 a_i b(x) \dx
               + \int_{I_i} b(x)^2 \dx \Bigg)
           \end{aligned}\\
&= \sum_{i \in [n]} \Bigg( \begin{aligned}[t]
               &\frac{a_i^2}{n}
               - 2 \frac{a_i}{\sqrt{n}} \frac{b_i}{\sqrt{m}}
               + \int_{I_i} b(x)^2 \dx \Bigg)
           \end{aligned}\\
        &\ge \sum_{i \in [n]} \Bigg( \begin{aligned}[t]
               &\frac{a_i^2}{n}- 2 \frac{a_i}{\sqrt{n}} \frac{b_i}{\sqrt{m}}
               + n \Big(\int_{I_i} b(x) \dx \Big)^2 \Bigg)
           \end{aligned} \tag*{\vspace{4em}(C.-S.)}\\
        &= \sum_{i \in [n]} \Bigg( \begin{aligned}[t]
               &\frac{a_i^2}{n}- 2 \frac{a_i}{\sqrt{n}} \frac{b_i}{\sqrt{m}}
                + \frac{b_i^2}{m}\Bigg)
           \end{aligned}\\
        &= \sum_{i \in [n]} \Big(\frac{1}{\sqrt{nm}} (m \cdot AX_Sv)_i - \frac{1}{\sqrt{nm}} (n \cdot X_S B v)_i \Big)^2\\
        &= \normT{\frac{1}{\sqrt{nm}} (m \cdot AX_S - n \cdot X_S B) v}^2.
    \end{align*}
    Combining this yields
    \begin{align*}
        \neiDistSpec(G, H)
        &= \inf_{X \in \mathcal{X}(G, H)} \sup_{\substack{v \in \R^m,\\ \normT{v} \le 1}} \normT{\frac{1}{\sqrt{nm}}(m \cdot A X - n \cdot X B)v}\\
        &\le \inf_{S \in \markov} \sup_{\substack{v \in \R^m,\\ \normT{v} \le 1}} \normT{\frac{1}{\sqrt{nm}}(m \cdot A X_S - n \cdot X_S B)v}\\
        &\le \inf_{S \in \markov} \sup_{\substack{v \in \R^m,\\ \normT{v} \le 1}} \normT{(T_{W_G} \circ S - S \circ T_{W_H}) g_v}\\
        &\le \inf_{S \in \markov} \sup_{\substack{\normT{g} \le 1}} \normT{(T_{W_G} \circ S - S \circ T_{W_H}) g}\\
        &\le \neiDistROp(W_G, W_H).
    \end{align*}
\end{proof}

For the path distance,
we have to verify that a signed fractional overlay
can be turned into a signed Markov operator (\Cref{le:signedFracOverlayToSignedMarkov})
and vice versa (\Cref{le:signedMarkovToSignedFracOverlay}).
Then, the proof of \Cref{le:steppingFunctionPathDistSpec}
is essentially analogous to the one of the second equality
of \Cref{le:steppingFunctionNeiDist},
which is why we omit it.

\begin{lemma}
    \label{le:signedFracOverlayToSignedMarkov}
    For a signed fractional overlay $X \in \R^{n \times m}$,
    the operator $T_X$ is a signed Markov operator with Hilbert adjoint $T_X^* = T_{X^T}$.
\end{lemma}
\begin{proof}
    We verify that $T_X$ is an $L_2$-contraction;
    the remaining part of the proof is the same as the proof of \Cref{le:fracOverlayToMarkov}.
    For an $f \in \LT$, we have
    \begin{align*}
        \normT{T_X f}^2&= \sum_{i \in [n]} \int_{I_i} \Bigg( \sum_{j \in [m]} \int_{J_j} nm \cdot X_{ij} \cdot f(y) \dy \Bigg)^2 \dx\\
        &= n m^2 \cdot \sum_{i \in [n]} \Bigg( \sum_{j \in [m]} X_{ij} \cdot \int_{J_j} f(y) \dy \Bigg)^2\\
        &\le n m^2 \cdot \frac{1}{nm} \cdot \sum_{j \in [m]} \left(\int_{J_j} f(y) \dy\right)^2 \tag*{($\normT{Xv}^2 \le \frac{1}{nm} \normT{v}^2$ for every $v \in \R^m$)}\\
        &\le n m^2 \cdot \frac{1}{nm} \cdot \sum_{j \in [m]} \frac{1}{m} \int_{J_j} f(y)^2 \dy \tag*{(Cauchy-Schwarz)}\\
        &= \normT{f}^2.
    \end{align*}
\end{proof}

\begin{lemma}
    \label{le:signedMarkovToSignedFracOverlay}
    For a signed Markov operator $S \colon \LT \to \LT$,
    the matrix $X_S$ is a signed fractional overlay of $G$ and $H$.
\end{lemma}
\begin{proof}
    We have $X_S \in \R^{n \times m}$ and
    verify that $\normT{X_S v} \le {\normT{v}}/{\sqrt{nm}}$
    for every $v \in \R^{m}$.
    The remaining part of the proof is the same as the proof of \Cref{le:markovToFracOverlay}.
    For $v \in \R^m$, we have
    \begin{align*}
        \normT{X_S v}^2&= \sum_{i \in [n]} \Bigg( \sum_{j \in [m]} v_j \cdot \int_{I_i} S(\cFun_{J_j})(x) \dx \Bigg)^2\\
        &= \sum_{i \in [n]} \Bigg( \int_{I_i} S \Big(\sum_{j \in [m]} v_j \cdot \cFun_{J_j}\Big)(x) \dx \Bigg)^2\\
        &\le \sum_{i \in [n]} \frac{1}{n} \int_{I_i} S \Big(\sum_{j \in [m]} v_j \cdot \cFun_{J_j}\Big)(x)^2 \dx \tag*{(Cauchy-Schwarz)}\\
&\le \frac{1}{n} \int_{[0,1]} \Big(\sum_{j \in [m]} v_j \cdot \cFun_{J_j} (x) \Big)^2 \dx \tag*{($\normT{Sf}^2 \le \normT{f}^2$ for every $f \in \LT$)}\\
        &= \frac{1}{n} \sum_{j \in [m]} \int_{J_j} v_j^2 \dx\\
&= \frac{1}{nm} \normT{v}^2.
    \end{align*}
\end{proof}

\subsection{Proof of \Cref{le:neiDistPseudoMetric} and \Cref{le:pathDistTPseudoMetric} (Pseudometrics)}
\label{subsec:missingPseudoMetricProof}

\begin{proof}[Proof of \Cref{le:neiDistPseudoMetric}]
    First, let $U \in \graphons$ be a graphon.
    Since the identity operator is a Markov operator,
    we immediately get $\neiDist(U, U) = \neiDistROp(U, U) = 0$.
    Second, let $U, W \in \graphons$ be graphons.
    Then, $(T_U \circ S - S \circ T_W)^* = S^* \circ T_U - T_W \circ S^*$
    since $T_U$ and $T_W$ are self-adjoint.
    Moreover, $S^*$ is a Markov operator by \Cref{le:markovClosure}.
    Thus,
    \begin{align*}
        \neiDist(U, W) &= \inf_{S \in \markov} \supFG \lvert \langle f, (T_U \circ S - S \circ T_W) g \rangle \rvert\\
&= \inf_{S \in \markov} \supFG \lvert \langle (S^* \circ T_U - T_W \circ S^*) f, g \rangle \rvert\\
        &= \inf_{S \in \markov} \supFG \lvert \langle g, (T_W \circ S^* - S^* \circ T_U) f \rangle \rvert \tag*{(Linearity and symmetry)}\\
        &= \inf_{S \in \markov} \supFG \lvert \langle g, (T_W \circ S - S \circ T_U) f \rangle \rvert\tag{$S^*$ is Markov and $S^{**} = S$}\\
        &= \neiDist(W, U).
    \end{align*}
    For $\neiDistROp$, proving symmetry is analogous as the operator norm is invariant
    under Hilbert adjoints.

    Let $U, V, W \in \graphons$ be graphons.
    For all Markov operators $S_1, S_2 \colon \LT \to \LT$,
    their composition $S_1 \circ S_2$ is also a Markov operator by \Cref{le:markovClosure}
    and
    \begin{align*}
        \neiDist(U, W) &= \inf_{S \in \markov} \supFG \lvert \langle f, (T_U \circ S - S \circ T_W) g \rangle \rvert\\
        &\le \supFG \lvert \langle f, (T_U \circ S_1 \circ S_2 - S_1 \circ S_2 \circ T_W) g \rangle \rvert\\
&= \supFG \lvert \langle f, ((T_U \circ S_1 - S_1 \circ T_V) \circ S_2 + S_1 \circ (T_V \circ S_2 - S_2 \circ T_W)) g \rangle \rvert\\
        &\le \supFG \lvert \langle f, ((T_U \circ S_1 - S_1 \circ T_V) \circ S_2) g \rangle \rvert\\
        &\qquad+ \supFG \lvert \langle f, (S_1 \circ (T_V \circ S_2 - S_2 \circ T_W)) g \rangle \rvert.
    \end{align*}
    For measurable functions $f, g \colon [0,1] \to [0,1]$,
    we have that $S_1^* f$ and $S_2 g$ again are measurable functions $[0,1] \to [0,1]$
    by \Cref{le:markovClosure} (in the case of $S_1^*$)
    and \Cref{le:markovContraction}.
    Thus,
    \begin{align*}
        &\supFG \lvert \langle f, ((T_U \circ S_1 - S_1 \circ T_V) \circ S_2) g \rangle \rvert\\
        \le {}&\supFG \lvert \langle f, (T_U \circ S_1 - S_1 \circ T_V) g \rangle \rvert
    \end{align*}
    and
    \begin{align*}
        &\supFG \lvert \langle f, (S_1 \circ (T_V \circ S_2 - S_2 \circ T_W)) g \rangle \rvert\\
        = {}&\supFG \lvert \langle S_1^* f, (T_V \circ S_2 - S_2 \circ T_W) g \rangle \rvert\\
        \le {}&\supFG \lvert \langle f, (T_V \circ S_2 - S_2 \circ T_W) g \rangle \rvert
    \end{align*}
    Hence, $\neiDist(U, W) \le \neiDist(U, V) + \neiDist(V, W)$.
    For $\neiDistROp$, the proof is analogous using the sub-multiplicativity
    of the operator norm and \Cref{le:markovContraction}, the fact that a Markov operator is a contraction,
    cf.\ also the proof of \Cref{le:pathDistTPseudoMetric}.
\end{proof}

\begin{proof}[Proof of \Cref{le:pathDistTPseudoMetric}]
    First, for a graphon $U \in \graphons$, we immediately get $\pathDistROp(U, U) = 0$
    since the identity operator is a signed Markov operator.
    Second, let $U, W \in \graphons$ be graphons.
    Then, $(T_U \circ S - S \circ T_W)^* = S^* \circ T_U - T_W \circ S^*$
    since $T_U$ and $T_W$ are self adjoint.
    The operator norm is invariant under taking the Hilbert adjoint, and we get
    \begin{align*}
        \pathDistROp(U, W)&= \inf_{S \in \signedMarkovs} \normRTT{T_U \circ S - S \circ T_W}\\
&= \inf_{S \in \signedMarkovs} \normRTT{S^* \circ T_U - T_W \circ S^*}\\
        &= \inf_{S \in \signedMarkovs} \normRTT{T_W \circ S^* - S^* \circ T_U}\\
        &= \inf_{S \in \signedMarkovs} \normRTT{T_W \circ S - S \circ T_U} \tag{$S^*$ is signed Markov and $S^{**} = S$}\\
        &= \pathDistROp(W, U).
    \end{align*}
    Third, let $U, V, W \in \graphons$ be graphons.
    For all signed Markov operators $S_1, S_2 \colon \LT \to \LT$,
    their composition $S_1 \circ S_2$ is also a signed Markov operator,
    and
    \begin{align*}
        \pathDistROp(U, W)&= \inf_{S \in \markov} \normRTT{T_U \circ S - S \circ T_W}\\
        &\le \normRTT{T_U \circ S_1 \circ S_2 - S_1 \circ S_2 \circ T_W}\\
        &\le \normRTT{T_U \circ S_1 \circ S_2 - S_1 \circ T_V \circ S_2 + S_1 \circ T_V \circ S_2 - S_1 \circ S_2 \circ T_W}\\
        &= \normRTT{(T_U \circ S_1 - S_1 \circ T_V) \circ S_2 + S_1 \circ (T_V \circ S_2 - S_2 \circ T_W)}\\
        &\le \normRTT{(T_U \circ S_1 - S_1 \circ T_V) \circ S_2} + \normT{S_1 \circ (T_V \circ S_2 - S_2 \circ T_W)}\\
        &\le \normRTT{T_U \circ S_1 - S_1 \circ T_V} \normT{S_2} + \normT{S_1} \normT{T_V \circ S_2 - S_2 \circ T_W} \tag*{(sub-mult.)}\\
        &\le \normRTT{T_U \circ S_1 - S_1 \circ T_V} + \normT{T_V \circ S_2 - S_2 \circ T_W}. \tag*{($S_1, S_2$ contractions)}
    \end{align*}
    Thus, $\pathDistROp(U, W) \le \pathDistROp(U, V) + \pathDistROp(V, W)$.
\end{proof}

\subsection{Proof of \Cref{le:treeDistInequalitiesSimple} and \Cref{le:neiDistZero} (Tree Distance Zero)}
\label{subsec:treeDistanceNorms}

\Cref{le:treeDistInequalitiesSimple} is a special case of the following
\Cref{le:treeDistInequalities}.
Note that
a Markov operator $S \colon \LT \to \LT$
uniquely extends to a Markov operator $S_1 \colon \LO \to \LO$, which
restricts to a Markov operator $S_p \colon \Lp \to \Lp$ for any $1 \le p \le \infty$,
cf.\ \Cref{th:markovRestriction}.
For graphons $U, W \in \graphons$, we can view $T_U$ and $T_W$
as operator $T_U, T_W \colon \LO \to \LI$.
Hence, we can view $T_U \circ S - S \circ T_W$ as an operator $\LO \to \LI$,
and letting
$\neiDistRpq(U, W) \coloneqq \inf_{S \in \markov} \normRpq{T_U \circ S - S \circ T_W}$
and
$\neiDistpq(U, W) \coloneqq \inf_{S \in \markov} \normpq{T_U \circ S - S \circ T_W}$
for all $U, W \in \graphons$ is well-defined for all $1 \le p,q \le \infty$;
the Riesz-Thorin Interpolation Theorem makes the detour via complex
$\Lp$ spaces necessary.
For the proof of \Cref{le:neiDistZero}, another variant of the tree distance is helpful, and we let
\begin{equation*}
    \neiDistEu(U, W) \coloneqq \inf_{S \in \markov} \sup_{\normT{f} \le 1, \normT{g} \le 1} \lvert \langle f, (T_U \circ S - S \circ T_W) g \rangle \rvert
\end{equation*}
for all $U, W \in \graphons$.
As with the cut distance, one can prove that these variants
of the tree distance yield the same topology, cf.\
\cite[Lemma $8.11$]{Lovasz2012}, and \cite[E.$2$ and E.$3$]{Janson2013}, and  in particular, \Cref{le:normsForGraphons}.

\begin{lemma}
    \label{le:neiDistCutEuOp}
    \label{le:neiDistReal}
    \label{le:neiDistOpIO}
    \label{le:treeDistInequalities}
    We have
    \begin{enumerate}
        \item $\neiDist \le \neiDistEu \le \neiDistROp \le \neiDistOp \le \sqrt{2} (\neiDistIO)^{1/2}$, \label{le:treeInequalities:one}
        \item $\neiDistIO \le 2 \neiDistIOR$, \label{le:treeInequalities:two}
        \item $\neiDistIOR \le 4 \neiDist$, and \label{le:treeInequalities:three}
        \item $\neiDistROp \le \neiDistOp \le 2 \neiDistROp$. \label{le:treeInequalities:four}
    \end{enumerate}
\end{lemma}
\begin{proof}[Proof of \Cref{le:treeDistInequalities}]
    Let $U,W \in \graphons$ be graphons.

    $(\ref{le:treeInequalities:one})$:
    For a measurable function $f \colon [0,1] \to [0,1]$,
    we have $\normT{f} \le \normI{f} \le 1$ and trivially get
    $\neiDist(U, W) \le \neiDistEu(U, W)$.
    Furthermore, for an operator $T \colon \LT \to \LT$, we have
    \begin{equation*}
        \sup_{\normT{f} \le 1, \normT{g} \le 1}\left\lvert \langle f, T g \rangle \right\rvert
        \le \sup_{\normT{f} \le 1, \normT{g} \le 1} \normT{f} \normT{\overline{T g}}
        \le \sup_{\normT{g} \le 1} \normT{T g}
        = \normRTT{T}
        \le \normTT{T}
    \end{equation*}
    by the Cauchy-Schwarz inequality,
    which that $\neiDistEu(U, W) \le \neiDistROp(U, W) \le \neiDistOp(U, W)$.
    The last inequality is a consequence of the Riesz-Thorin Interpolation Theorem.
    Let $S \colon \LT \to \LT$ be a Markov operator,
    and let $S_1$ be its unique extension to a Markov operator
    $S_1 \colon \LO \to \LO$, cf.\ \Cref{th:markovRestriction}.
    Let $p = q = 2$, $p_0 = 1$, $q_0 = \infty$, $p_1 = \infty$, $q_1 = 1$,
    and $\theta = 1/2$.
    Then,
    \begin{align*}
        &\frac{1}{p} = \frac{1-\theta}{p_0} + \frac{\theta}{p_1}& &\text{and}&
        &\frac{1}{q} = \frac{1-\theta}{q_0} + \frac{\theta}{q_1},&
    \end{align*}
    that is, the Riesz-Thorin Interpolation Theorem, \Cref{th:rieszThorin}, is applicable,
    and we get
    \begin{equation*}
        \normTT{T_U \circ S_1 - S_1 \circ T_W}
        \le \normOI{T_U \circ S_1 - S_1 \circ T_W}^{1/2}\normIO{T_U \circ S_1 - S_1 \circ T_W}^{1/2}.
    \end{equation*}
    As
    \begin{align*}
        &\normOI{T_U \circ S_1 - S_1 \circ T_W}\\
        \le {}&\normOI{T_U \circ S_1} + \normOI{S_1 \circ T_W}\\
        \le {}&\normOI{T_U} \normOO{S_1} + \normII{S_1} \normOI{T_W} \tag{sub-multiplicativity}\\
        \le {}&\normOI{T_U} + \normOI{T_W} \tag{\Cref{le:markovContraction}}\\
        \le {}&2, \tag{$U, W$ graphons}
    \end{align*}
    we get $\normTT{T_U \circ S_1 - S_1 \circ T_W} \le \sqrt{2} \normIO{T_U \circ S_1 - S_1 \circ T_W}^{1/2}$ and, hence,
    \begin{equation*}
        \normTT{T_U \circ S - S \circ T_W} \le \sqrt{2} \normIO{T_U \circ S - S \circ T_W}^{1/2}.
    \end{equation*}
    This means that $\neiDistOp(U, W) \le \sqrt{2} \neiDistIO(U, W)^{1/2}$.

    $(\ref{le:treeInequalities:two})$:
    Let $S \colon \LT \to \LT$ be a Markov operator.
    Then,
    \begin{align*}
        &\normIO{T_U \circ S - S \circ T_W}\\
        = {}&\sup_{\substack{g \colon [0,1] \to \C,\\\normI{g} \le 1}} \normO{(T_U \circ S - S \circ T_W) (\Re g + \iu \Im g)}\\
        = {}&\sup_{\substack{g \colon [0,1] \to \C,\\\normI{g} \le 1}} \normO{(T_U \circ S - S \circ T_W) (\Re g) + \iu (T_U \circ S - S \circ T_W) (\Im g)}\\
        \le {}&\sup_{\substack{g \colon [0,1] \to \C,\\\normI{g} \le 1}} \big(\normO{(T_U \circ S - S \circ T_W) (\Re g)} + \normO{(T_U \circ S - S \circ T_W) (\Im g)}\big)\\
        \le {}&2 \cdot \sup_{\substack{g \colon [0,1] \to \R,\\\normI{g} \le 1}} \normO{(T_U \circ S - S \circ T_W) g},\tag*{($\normI{\Re g}, \normI{\Im g} \le \normI{g}$)}
    \end{align*}
    which yields $\neiDistIO(U, W) \le 2 \neiDistIOR(U, W)$.

    $(\ref{le:treeInequalities:three})$:
    Let $S \colon \LT \to \LT$ be
    a Markov operator.
    The one-dimensional cut norm coincides
    with the $L_1$-norm \cite[Remark $4.4$]{Janson2013}, i.e.,
    for a function $g \in \LO$, we have
$\normO{g} = \sup_{\normI{f} \le 1} \big\lvert \int_{[0,1]} f(x) g(x) \dx \big\rvert$.
and get
    \begin{align*}
        &\sup_{\substack{g \colon [0,1] \to \R,\\\normI{g} \le 1}} \normO{(T_U \circ S - S \circ T_W) g}\\
        = {}&\sup_{f, g \colon [0,1] \to [-1,1]} \lvert\langle f, (T_U \circ S - S \circ T_W) g \rangle\rvert\\
        = {}&\sup_{f,f', g,g' \colon [0,1] \to [0,1]} \lvert\langle f - f', (T_U \circ S - S \circ T_W) (g - g') \rangle\rvert\\
        \le {}&4 \cdot \sup_{f, g \colon [0,1] \to [0,1]} \lvert\langle f, (T_U \circ S - S \circ T_W) g \rangle\rvert\\
    \end{align*}
    since, for $T \coloneqq T_U \circ S - S \circ T_W$,
    \begin{equation*}
        \langle f - f', T (g - g') \rangle = \langle f, T g \rangle - \langle f', T g \rangle - \langle f, T g' \rangle + \langle f', T g' \rangle.
    \end{equation*}
    Hence, $\neiDistIOR(U, W) \le 4 \neiDist(U, W)$.

    $(\ref{le:treeInequalities:four})$:
    The first inequality is trivial.
    To prove the second, let $S \in \markov$ be a Markov operator
    and $g \colon [0,1] \to \C$ be a function in $\LT$ with $\normT{g} \le 1$.
    Then, $\Re g, \Im g \in \LT$ with $\normT{\Re g}, \normT{\Im g} \le 1$.
    Moreover,
    \begin{align*}
        \normT{(T_U \circ S - S \circ T_W) g}&= \normT{(T_U \circ S - S \circ T_W) (\Re g + \iu \Im g)}\\
        &\le \normT{(T_U \circ S - S \circ T_W) (\Re g)} + \normT{\iu (T_U \circ S - S \circ T_W) (\Im g)}\\
        &= \normT{(T_U \circ S - S \circ T_W) (\Re g)} + \normT{(T_U \circ S - S \circ T_W) (\Im g)}.
    \end{align*}
    and, hence,
    \begin{align*}
        \normTT{T_U \circ S - S \circ T_W}&= \sup_{\substack{g \colon [0,1] \to \C,\\ \normT{g} \le 1}} \normT{(T_U \circ S - S \circ T_W) g}\\
        &\le 2 \sup_{\substack{g \colon [0,1] \to \R,\\ \normT{g} \le 1}} \normT{(T_U \circ S - S \circ T_W) g}\\
        &= 2 \normRTT{T_U \circ S - S \circ T_W}.
    \end{align*}
\end{proof}

\label{subsec:infimumAttained}

Recall that the set of Markov operators is compact
in the weak operator topology, cf.\ \Cref{th:markovCompact}.
This is the reason why we consider $\neiDistEu$
instead of $\neiDistROp$ in the following lemma;
for it, compactness in the weak operator topology suffices
to prove that the infimum in its definition is attained.

\begin{lemma}
    \label{le:neiDistInfMin}
    The infimum in the definition of $\neiDistEu$ is attained.
\end{lemma}
\begin{proof}[Proof of \Cref{le:neiDistInfMin}]
    Let $U, W \in \graphons$ be graphons.
    As the set $\markov$ is compact in the weak operator topology by
    \Cref{th:markovCompact},
    it suffices to prove that the function $h$ defined by
    \begin{align*}
        h(S) &= \sup_{\normT{f} \le 1, \normT{g} \le 1}\big\lvert \langle f, (T_U \circ S) g \rangle - \langle f, (S \circ T_W) g \rangle \big\rvert
    \end{align*}
    is lower semi-continuous.
    To this end, let $\{S_i\}_{i \in I}$ be a net of Markov operators converging
    to a Markov operator $S^* \in \markov$ in the weak operator topology,
    i.e.,
    we have $\langle f, S_i g \rangle \rightarrow \langle f, S^* g \rangle$
    for all $f,g \in \LT$.
    We have to show that $\liminf_{i} h(S_i) \ge h(S^*)$.
    Let $f,g \in \LT$ with $\normT{f}, \normT{g} \le 1$.
    We have
    \begin{align*}
        \langle f, (T_U \circ S_i) g \rangle = \langle T_U f, S_i g \rangle \rightarrow \langle  T_U f, S^* g \rangle
        =\langle f, (T_U \circ S^*) g)\rangle,
    \end{align*}
    where we used that $T_U$ is self-adjoint, and
    \begin{align*}
        \langle f, (S_i \circ T_W) g \rangle = \langle f, S_i (T_W g) \rangle \rightarrow \langle f, S^* (T_W g) \rangle = \langle f, (S^* \circ T_W) g \rangle.
    \end{align*}
    Together, this yields
    \begin{align*}
        \big\lvert \langle f, (T_U \circ S_i) g \rangle - \langle f, (S_i \circ T_W) g \rangle \big\rvert \rightarrow \big\lvert \langle f, (T_U \circ S^*) g \rangle - \langle f, (S^* \circ T_W) g \rangle \big\rvert,
    \end{align*}
    which gives us
    \begin{align*}
        \liminf_{i} h(S_i)&\ge \liminf_{i} \big\lvert \langle f, (T_U \circ S_i) g \rangle - \langle f, (S_i \circ T_W) g \rangle \big\rvert\\
        &= \big\lvert \langle f, (T_U \circ S^*) g \rangle - \langle f, (S^* \circ T_W) g \rangle \big\rvert.
    \end{align*}
    Since this holds for all $f, g \in \LT$ with $\normT{f}, \normT{g} \le 1$,
    we get $\liminf_{i} h(S_i) \ge h(S^*)$
    by definition of the supremum.
\end{proof}

\label{subsec:distanceZero}

With \Cref{le:treeDistInequalities} and \Cref{le:neiDistInfMin},
proving \Cref{le:neiDistZero} is easy.

\begin{proof}[Proof of \Cref{le:neiDistZero}]
    Let $U, W \in \graphons$ be graphons.
    If $t(T, U) = t(T, W)$ for every tree $T$,
    then, by \Cref{th:fracIsoGraphons}, there is a Markov operator
    $S \in \markov$ such that $T_U \circ S = S \circ T_W$,
    which directly yields $\neiDist(U, W) = 0$ by the definition of $\neiDist$.

    For the other direction, assume that $\neiDist(U, W) = 0$.
    Then, by \Cref{le:treeDistInequalities}, we have $\neiDistEu(U, W) = 0$.
    By \Cref{le:neiDistInfMin}, there is a Markov operator
    $S \in \markov$ such that
    $\langle f, (T_U \circ S) g \rangle = \langle f, (S \circ T_W) g \rangle$
    for all $f,g \in \LT$ with $\normT{f}, \normT{g} \le 1$.
    Since this, in particular, holds for $f = g$ and, since we can just normalize
    an arbitrary $g \in \LT$, linearity of the operators and the inner product yields
    $\langle g, (T_U \circ S) g \rangle = \langle g, (S \circ T_W) g \rangle$
    for every $g \in \LT$.
    Hence, \Cref{le:operatorIsZero} yields $T_U \circ S = S \circ T_W$, and
    by \Cref{th:fracIsoGraphons}, we have
    $t(T, U) = t(T, W)$ for every tree $T$.
\end{proof}

\subsection{Proof of \Cref{le:neiDistLeCutDist} (\texorpdfstring{$\neiDist \le \dist$}{Tree and Cut Distance})}
\label{subsec:treeDistLeCutDist}

\begin{proof}[Proof of \Cref{le:neiDistLeCutDist}]
    Let $U, W \in \graphons$ be graphons, and let
    $\varphi \in \measPres$ be an invertible measure-preserving map.
    The Koopman operator $T_\varphi$ of $\varphi$ is a Markov operator,
    and we observe that
    \begin{align*}
        T_U \circ T_\varphi - T_\varphi \circ T_W&= (T_U \circ T_\varphi - T_\varphi \circ T_W) \circ T_{\varphi^{-1}} \circ T_\varphi\\
        &= (T_U - T_\varphi \circ T_W \circ T_\varphi^{-1}) \circ T_\varphi\\
        &= (T_U - T_{W^\varphi}) \circ T_\varphi\\
        &= T_{U-W^\varphi} \circ T_\varphi.
    \end{align*}
    Then, we get
    \begin{align*}
        &\supFG \lvert \langle f, (T_U \circ T_\varphi - T_\varphi \circ T_W) g \rangle \rvert\\
= {}&\supFG \lvert \langle f, T_{U-W^\varphi} (T_\varphi g) \rangle \rvert\\
        = {}&\supFG \lvert \langle f, T_{U-W^\varphi} g \rangle \rvert. \tag{$\varphi$ measure preserving}
    \end{align*}
    Thus,
    \begin{align*}
        \neiDist(U, W) &= \inf_{S \in \markov} \supFG \lvert \langle f, (T_U \circ S - S \circ T_W) g \rangle \rvert\\
        &\le \inf_{\varphi \in \measPres} \supFG \lvert \langle f, (T_U \circ T_\varphi - T_\varphi \circ T_W) g \rangle \rvert\\
        &= \inf_{\varphi \in \measPres} \supFG \lvert \langle f, T_{U-W^\varphi} g \rangle \rvert\\
        &= \dist(U, W).
    \end{align*}
\end{proof}

\subsection{Proof of \Cref{th:pathHomsGraphons} (Path Densities)}
\label{subsec:pathDensitiesProof}

The following lemma presents a generalization
of the interpolation technique of Dell, Grohe, and Rattan \cite[Lemma $10$]{Dell2018}
from finite sums to convergent series
and is needed for the proof of \Cref{th:pathHomsGraphons}.

\begin{lemma}[{\cite[Proposition A$.21$]{Lovasz2012}}]
    \label{le:interpolationLemma}
    Let $a_i, b_i, c_i, d_i$ be sequences of non-zero real numbers
    such that $b_i \neq b_j$ and $d_i \neq d_j$ for $i \neq j$.
    Assume that there is a $k_0 \ge 0$ such that, for every $k \ge k_0$,
    the sums $\sum_{i = 1}^\infty a_i b_i^k$ and $\sum_{i = 1}^\infty c_i d_i^k$
    are convergent and equal.
    Then, there is a permutation $\pi \colon \N \to \N$ such that
    $a_i = c_{\pi(i)}$ and $b_i = d_{\pi(i)}$ for every $i \ge 0$.
\end{lemma}

Note that, if $k_0 = 0$, we can also allow the number zero to appear
in the sequences $b_i$ and $d_i$;
equality of the two series for $k = 0$ directly implies that
the coefficients of the zeros are the same.

\begin{proof}[Proof of \Cref{th:pathHomsGraphons}]
    First, assume that there is some operator $S \colon \LT \to \LT$
    with $S \allOne = \allOne$ and $S^* \allOne = \allOne$
    such that $T_U \circ S = S \circ T_W$.
    Then, induction yields $T_U^\ell \circ S = S \circ T_W^\ell$ for every $\ell \ge 0$.
    Hence,
    \begin{equation*}
        \langle \allOne, T_U^\ell \allOne \rangle = \langle \allOne, (T_U^\ell \circ S) \allOne \rangle = \langle \allOne, (S \circ T_W^\ell) \allOne \rangle = \langle S^* \allOne, T_W^\ell \allOne \rangle = \langle \allOne, T_W^\ell \allOne \rangle.
    \end{equation*}

    For the backward direction, assume that $t(P_\ell, U) = t(P_\ell, W)$ for every $\ell \ge 0$.
    Let $\{f'_i\}_{i \in \N}$ and $(\lambda'_i)_{i \in \N}$
    be the orthonormal basis of the $\LT$
    consisting of eigenfunctions of $T_U$ and the corresponding
    sequence of real eigenvalues obtained from the Spectral Theorem for $T_U$.
    By definition of an orthonormal basis, we have
    $\allOne = \sum_{i \in \N} \langle \allOne, f'_i \rangle f'_i$.
    We call an eigenvalue $\lambda$ from $(\lambda'_i)_{i \in \N}$ \textit{useful}
    if the finite sum $\sum_{i \in \N, \lambda'_i = \lambda} \langle \allOne, f_i \rangle f_i$
    is non-zero, i.e., one of the eigenfunctions corresponding to $\lambda$ is
    not orthogonal to $\allOne$.
    Let $\{\lambda_i\}_{i \ge 0} \subseteq \{\lambda'_1, \lambda'_2, \dots\}$
    be the set of these (pairwise distinct) useful eigenvalues
    and let $f_i \coloneqq \sum_{j \in \N, \lambda'_j = \lambda_i} \langle \allOne, f'_j \rangle f'_j$
    for $i \ge 0$.
    Then, $\allOne = \sum_{i \ge 0} f_i$, where $f_i$ is an eigenfunction of $T_U$ with eigenvalue
    $\lambda_i$, and the set $\{f_i\}_{i \ge 0}$ of these eigenfunctions is orthogonal.
    Note that the set $\{\lambda_i\}_{i \ge 0}$ may be finite;
    in terms of notation, we do not treat this case differently.
    In the same way, apply the Spectral Theorem to $T_W$ to obtain
    another orthonormal basis and sequence of eigenvalues and define the useful
    eigenvalues $\{\mu_i\}_{i \ge 0}$ and the functions $\{g_i\}_{i \ge 0}$ analogously.

    Then, since we have $T_U(f_i) = \lambda_i f_i$ for every $i \ge 0$,
    we get
    \begin{align*}
        \langle \allOne, T_U^\ell \allOne \rangle &= \langle \allOne, T_U^\ell \big(\sum_{i \ge 0} f_i \big)\rangle\\
        &= \langle \allOne, \sum_{i \ge 0} \lambda_i^\ell f_i \rangle \tag*{($T_U$ linear and continuous)}\\
        &= \langle \sum_{i \ge 0} f_i, \sum_{i \ge 0} \lambda_i^\ell f_i \rangle\\
        &= \sum_{i \ge 0} \lambda_i^\ell \langle f_i, f_i \rangle \tag*{($\langle \cdot, \cdot \rangle$ bilinear and continuous)}\\
        &= \sum_{i \ge 0} \normT{f_i}^2 \lambda_i^\ell
    \end{align*}
    for every $\ell \ge 0$.
    Analogously, we get $\langle \allOne, T_W^\ell \allOne \rangle = \sum_{i \ge 0} \normT{g_i}^2 \mu_i^\ell$,
    and the assumption can be formulated as
$\sum_{i \ge 0} \normT{f_i}^2 \lambda_i^\ell = \sum_{i \ge 0} \normT{g_i}^2 \mu_i^\ell$
for every $\ell \ge 0$.

    We argue that \Cref{le:interpolationLemma} is applicable.
    If both sets $\{\lambda_i\}_{i \ge 0}$ and $\{\mu_i\}_{i \ge 0}$
    are infinite, this is clear.
    If both sets are finite, the lemma also applies
    as we can simply append a sequence like $(2^{-i})_{i \ge i_0}$ for some $i_0 \ge 0$
    to both sequences.
    We argue that the remaining case,
    where one of the sets is finite while the other
    one is infinite, cannot occur.
    To this end, assume without loss of generality that
    $\{\lambda_i\}_{i \ge 0}$ is the finite set $\{\lambda_0, \dots, \lambda_n\}$.
    Then, the assumption reads as
    $\sum_{i = 0}^{n} \normT{f_i}^2 \lambda_i^\ell = \sum_{i = 0}^{\infty} \normT{g_i}^2 \mu_i^\ell$ for every $\ell \ge 0$, which implies that
    \begin{equation*}
        \sum_{i = 0}^{n} \normT{f_i}^2 \lambda_i^\ell + \sum_{i = 0}^{\infty} \normT{g_i}^2 \mu_i^\ell = \sum_{i = 0}^{\infty} 2  \normT{g_i}^2 \mu_i^\ell
    \end{equation*}
    for every $\ell \ge 0$.
    By combining the finite sum and the infinite series on the left-hand side,
    we are again in the situation of \Cref{le:interpolationLemma},
    where the sequences $a_i$ and $b_i$ for the left-hand side have finitely many elements
    of the form $\normT{f_i}^2$ and $\lambda_i$,
    finitely many elements of the form $\normT{f_i}^2 + \normT{g_j}^2$ and $\lambda_i = \mu_j$,
    and infinitely many elements of the form $\normT{g_i}^2$ and $\mu_i$,
    respectively.
    In contrast, the elements of the sequences $c_i$ and $d_i$ for the right-hand side
    are of the form $2\normT{g_i}^2$ and $\mu_i$, respectively.
    Hence, the resulting bijection has to map
    one of the infinitely many pairs of elements of the form $\normT{g_i}^2$ and $\mu_i$ to
    a pair $2\normT{g_j}^2$ and $\mu_j$.
    Then, $i = j$ since the $\{\mu_i\}_{i \ge 0}$ are pairwise distinct
    and the lemma guarantees that $\mu_i = \mu_j$.
    But, we have $\normT{g_i}^2 \neq 2 \normT{g_i}^2$, which contradicts the lemma.

    Now, \Cref{le:interpolationLemma}
    yields a permutation $\pi \colon \N \to \N$ such that
    $\lambda_i = \mu_{\pi(i)}$ and $\normT{f_i}^2 = \normT{g_{\pi(i)}}^2$
    for every $i \ge 0$.
    By relabeling, we can assume $\lambda_i = \mu_i$ and $\normT{f_i} = \normT{g_i}$
    for every $i \ge 0$.
    Note that, as the convergence in an orthonormal basis is unconditional by
    definition, this does not change the fact that we have $\allOne = \sum_{i \ge 0} g_i$.
    For a function $f \in \LT$, define
    \begin{equation*}
        S f \coloneqq \sum_{i \ge 0} {\big\langle f, \frac{g_i}{\normT{g_i}} \big\rangle} \frac{f_i}{\normT{f_i}}.
    \end{equation*}
    This actually defines a mapping $\LT \to \LT$:
    As $\{{f_i}/{\normT{f_i}}\}_{i \ge 0}$ is orthonormal,
    the Riesz-Fischer Theorem yields that the sum converges to a function in $\LT$
    if and only if we have $\sum_{i \ge 0} \left\lvert {\langle f, g_i/{\normT{g_i}} \rangle} \right\rvert^2 < \infty$.
    This however, follows immediately from Bessel's inequality as
    the set $\{{g_i}/{\normT{g_i}}\}_{i \ge 0}$ is also orthonormal, i.e.,
    we have
    $\sum_{i \ge 0} \left\lvert {\langle f, g_i/{\normT{g_i}\rangle}} \right\rvert^2 \le \normT{f}^2 < \infty$.
    The linearity of the inner product in its first argument yields that
    $S$ is linear.
    A closer analysis yields that
    \begin{align*}
        \normT{S f}^2
        &= \langle S f, S f \rangle\\
        &= \big\langle \sum_{i \ge 0} {\big\langle f, \frac{g_i}{\normT{g_i}} \big\rangle} \frac{f_i}{\normT{f_i}}, S f \big\rangle\\
        &= \sum_{i \ge 0} {\big\langle f, \frac{g_i}{\normT{g_i}} \big\rangle} \big\langle \frac{f_i}{\normT{f_i}}, S f \big\rangle \tag*{($\langle \cdot, \cdot \rangle$ bilinear and continuous)}\\
        &= \sum_{i \ge 0} {\big\langle f, \frac{g_i}{\normT{g_i}} \big\rangle} \sum_{j \ge 0} {\big\langle f, \frac{g_j}{\normT{g_j}} \big\rangle} \big\langle \frac{f_i}{\normT{f_i}}, \frac{f_j}{\normT{f_j}} \big\rangle \tag*{($\langle \cdot, \cdot \rangle$ bilinear and continuous)}\\
        &= \sum_{i \ge 0} {\big\langle f, \frac{g_i}{\normT{g_i}} \big\rangle} {\big\langle f, \frac{g_i}{\normT{g_i}} \big\rangle} \tag*{($\{f_i/\normT{f_i}\}_{i \ge 0}$ orthonormal)}\\
        &= \sum_{i \ge 0} \left \lvert {\big\langle f, \frac{g_i}{\normT{g_i}} \big\rangle} \right\rvert^2\\
        &\le \normT{f}^2 \tag*{(Bessel's inequality)}
    \end{align*}
    for every $f \in \LT$, i.e., $\normT{S f} \le \normT{f}$.
    Hence, $S \colon \LT \to \LT$ is not only a bounded linear operator but also a
    contraction.
    Moreover, we have
    \begin{align*}
        S \allOne
        &= \sum_{i \ge 0} {\big\langle \allOne, \frac{g_i}{\normT{g_i}} \big\rangle} \frac{f_i}{\normT{f_i}}\\
        &= \sum_{i \ge 0} \frac{\langle \allOne, g_i \rangle}{\normT{g_i}^2} f_i \tag*{($\normT{f_i} = \normT{g_i}$ for every $i \ge 0$)}\\
&= \sum_{i \ge 0} \frac{\langle g_i, g_i \rangle}{\normT{g_i}^2} f_i \tag*{($\{g_i\}_{i \ge 0}$ orthogonal, $\langle \cdot, \cdot \rangle$ continuous and linear)}\\
        &= \sum_{i \ge 0} f_i\\
        &= \allOne.
    \end{align*}
    It is easy to verify that the Hilbert adjoint $S^*$ of $S$ is given by
    \begin{equation*}
        S^* f = \sum_{i \ge 0} {\big\langle f, \frac{f_i}{\normT{f_i}} \big\rangle} \frac{g_i}{\normT{g_i}}
    \end{equation*}
    for every $f \in \LT$,
    and hence, by symmetry, we also have $S^* \allOne = \allOne$.
    Therefore, $S$ is a signed Markov operator.
    It remains to prove that $T_U \circ S = S \circ T_W$.
    We have
    \begin{align*}
        (T_U \circ S) f&= T_U \big(\sum_{i \ge 0} {\big\langle f, \frac{g_i}{\normT{g_i}} \big\rangle} \frac{f_i}{\normT{f_i}}\big)\\
        &= \sum_{i \ge 0} {\big\langle f, \frac{g_i}{\normT{g_i}} \big\rangle} \frac{T_U(f_i)}{\normT{f_i}} \tag*{($T_U$ linear and continuous)}\\
        &= \sum_{i \ge 0} {\big\langle f, \frac{g_i}{\normT{g_i}} \big\rangle} \frac{\lambda_i f_i}{\normT{f_i}}\\
        &= \sum_{i \ge 0} {\big\langle f, \frac{\mu_i g_i}{\normT{g_i}} \big\rangle} \frac{f_i}{\normT{f_i}} \tag*{($\lambda_i = \mu_i \in \R$ for every $i \ge 0$)}\\
        &= \sum_{i \ge 0} {\big\langle f, \frac{T_W(g_i)}{\normT{g_i}} \big\rangle} \frac{f_i}{\normT{f_i}}\\
        &= \sum_{i \ge 0} {\big\langle T_W f, \frac{g_i}{\normT{g_i}} \big\rangle} \frac{f_i}{\normT{f_i}} \tag*{($T_W$ self-adjoint)}\\
        &= (S \circ T_W) f
    \end{align*}
    for every $f \in \LT$.
\end{proof}

\subsection{Proof of \Cref{th:colRefApproxInv} (Approximate Inversion)}
\label{subsec:inverseApproximationProof}

Let us state the inversion result that \Cref{th:colRefApproxInv}
is based on.

\begin{theorem}[{\cite[Corollary 4]{KieferSchweitzerSelman15}}]
    \label{th:inversion}
    $\mathcal{I}^2_{C}$ admits linear time inversion on the class of graphs.
\end{theorem}
They show that, given
$\bar{s} \in \N^{m}$ and $M \in \N^{m \times m}$ such that
\begin{multicols}{2}
\begin{enumerate}
    \item $M_{ii} < s_i$ for every $i \in \numTo{m}$,
    \item $M_{ii} \cdot s_i$ is even for every $i \in \numTo{m}$,
    \item $M_{ij} \le s_j$ for all $i,j \in \numTo{m}$, and
    \item $M_{ij} \cdot s_i = M_{ji} \cdot s_j$ for all $i,j \in \numTo{m}$,
\end{enumerate}
\end{multicols}
\noindent one can construct a graph $G$ in linear time where
$V(G)$ can be partitioned into sets $C_1, \dots, C_m$ of sizes
$s_1, \dots, s_m$, respectively, such that
$G[C_i]$ is a $M_{ii}$-regular graph for every $i \in \numTo{m}$ and
$G[C_i \cup C_j]$ is a $(M_{ij}, M_{ji})$-biregular graph for all $i,j \in \numTo{m}$.
That is, these conditions, which are clearly necessary,
are also sufficient for such a graph to exist.

For a graph $G$ constructed from $\bar{s} \in \N^m$
and $M \in \N^{m \times m}$ via the criteria of \Cref{th:inversion},
the weighted graph $\ColG$ might not be isomorphic to
$G_{\bar{s},M} \coloneqq (\numTo{m}, ({s_i}/{\sum_{i} s_i})_{i}, ({M_{ij}}/{s_j})_{ij})$
as color refinement might compute a coarser partition than $C_1, \dots, C_m$.
The proof of \Cref{th:colRefApproxInv} proceeds
in three steps:
First, we round the vertex weights of the given weighted graph $H$.
\Cref{le:roundingVertexWeights} shows that this is possible
with an error of $1/n$ while using $n$ vertices.
Second, we round the edge weights of $H$.
Finally, we show that our resulting invariant is actually close to $H$.
Ideally, one would like to use the (scaled) identity matrix as a fractional overlay
for this. However, due to the rounded vertex weights,
we have to settle for a fractional overlay that is close to the identity matrix,
cf.\ \Cref{le:fracOverlay}.

\begin{lemma}
    \label{le:roundingVertexWeights}
    Let $\bar{\alpha} \in \Rnn^{m}$ such that $\sum_{i \in \numTo{m}} \alpha_i = 1$.
    For every $n \ge 1$, there is an $\bar{s} \in \N^n$ such that
    $\sum_{i \in \numTo{m}} s_i = n$ and $\lvert \frac{s_i}{n} - \alpha_i\rvert < \frac{1}{n}$
    for every $i \in \numTo{m}$.
\end{lemma}
\begin{proof}
    Clearly, the bound $\lvert \frac{s_i}{n} - \alpha_i \rvert < \frac{1}{n}$
    can be satisfied by setting $s_i \coloneqq \lfloor n \cdot \alpha_i \rfloor$
    or $s_i \coloneqq \lceil n \cdot \alpha_i \rceil$ for every $i \in \numTo{m}$.
    However, to also satisfy $\sum_{i \in \numTo{m}} s_i = n$, one has to
    choose correctly between these two alternatives.
    For $j = 1, \ldots, m$, we proceed as follows:
    If $\sum_{i \in \numTo{j}} (\frac{s_i}{n} - \frac{n \cdot \alpha_i}{n} ) > 0$,
    then we set $s_j \coloneqq \lfloor n \cdot \alpha_i \rfloor$.
    Otherwise, we set $s_j \coloneqq \lceil n \cdot \alpha_i \rceil$.
    A simple inductive argument yields that, for every $j \in \numTo{m}$, the invariant
    $\lvert \sum_{i \in \numTo{j}} (\frac{s_i}{n} - \frac{n \cdot \alpha_i}{n} ) \rvert < \frac{1}{n}$ is satisfied.
    In particular, we have $\lvert \sum_{i \in \numTo{m}} (\frac{s_i}{n} - \frac{n \cdot \alpha_i}{n} ) \rvert < \frac{1}{n}$.
    Since $\sum_{i \in \numTo{m}} \alpha_i = 1$, we get
    $\lvert \sum_{i \in \numTo{m}} \frac{s_i}{n} - 1 \rvert < \frac{1}{n}$,
    and by multiplying with $n$, also
    $\lvert \sum_{i \in \numTo{m}} s_i - n \rvert < 1$.
    Since $s_i \in \N$ for every $i \in \numTo{m}$, this implies
    $\sum_{i \in \numTo{m}} s_i = n$.
\end{proof}

\begin{lemma}
    \label{le:fracOverlay}
    Let $\bar{s} \in \Rnn^m$ and $\bar{t} \in \Rnn^n$ such that $\sum_{j = 1}^{m} s_j = \sum_{i = 1}^{n} t_i$.
    Then, there is an $X \in \Rnn^{m \times n}$ such that
    \begin{enumerate}
        \item $\sum_{j = 1}^{n} X_{ij} = s_i$ for every $i \in \numTo{m}$,
        \item $\sum_{i = 1}^{m} X_{ij} = t_j$ for every $j \in \numTo{n}$, and
        \item $X_{ii} = \min \{s_i, t_i\}$ for every $i \in \numTo{\min \{m, n\}}$.
    \end{enumerate}
\end{lemma}
\begin{proof}
    We prove the statement by induction on the total number of non-zero entries
    of $\bar{s}$ and $\bar{t}$.
    If $\bar{s}$ and $\bar{t}$ are all-zero vectors, then the desired $X$
    is obtained by choosing the all-zero matrix.
    Now, assume that $\bar{s}$ or $\bar{t}$ has a non-zero entry.
    Then, since $\sum_{j = 1}^{m} s_j = \sum_{i = 1}^{n} t_i$, both $\bar{s}$ and $\bar{t}$ have a non-zero entry.
    Since we can just transpose $X$ and swap the roles of $\bar{s}$ and $\bar{t}$,
    we may assume $m \ge n$ without loss of generality.

    Case $1$: There is no $k \in \numTo{n}$ such that $s_k > 0$ and $t_k > 0$.\\
    Let $k \in \numTo{m}$ such that $s_k > 0$ and let $\ell \in \numTo{n}$
    such that $t_\ell > 0$.
    Consider
    \begin{align*}
        &s_j' \coloneqq \begin{cases}
            s_j &\text{if } j \neq k,\\
            s_k - \min\{s_k, t_\ell\} &\text{if } j = k,
        \end{cases}&
        &\text{and}&
        &t_i' \coloneqq \begin{cases}
            t_i &\text{if } i \neq \ell,\\
            t_\ell - \min\{s_k, t_\ell\} &\text{if } i = \ell.
        \end{cases}&
    \end{align*}
    Then, $\sum_{j = 1}^{m} s'_j = \sum_{i = 1}^{n} t'_i$
    and, since $s'_k = 0$ or $t'_\ell = 0$, in total $\bar{s'}$ and $\bar{t'}$ have one less non-zero entry than $\bar{s}$ and $\bar{t}$.
    The induction hypothesis yields an $X' \in \Rnn^{m \times n}$ such that
    \begin{enumerate}
        \item $\sum_{j = 1}^{n} X'_{ij} = s'_i$ for every $i \in \numTo{m}$,
        \item $\sum_{i = 1}^{m} X'_{ij} = t'_j$ for every $j \in \numTo{n}$, and
        \item $X'_{jj} = \min \{s'_j, t'_j\}$ for every $j \in \numTo{n}$.
    \end{enumerate}
    Since $s'_k = 0$ or $t'_\ell = 0$, we have $X'_{k\ell} = 0$.
    Let $X$ be the matrix obtained from $X'$
    by replacing $X'_{k\ell}$ with $\min\{s_k, t_\ell\}$.
    By the case assumption, we have $k \neq \ell$
    and also $s'_j = 0$ or $t'_j = 0$ for every $j \in \numTo{n}$.
    Thus, $X_{jj} = X'_{jj} = \min\{s'_j, t'_j\} = 0 = \min\{s_j, t_j\}$ for every $j \in \numTo{n}$.
    Hence, $X$ has the desired properties.

    Case $2$: There is an $k \in \numTo{n}$ such that $s_k > 0$ and $t_k > 0$.\\
    We proceed as in the first case, where we choose $\ell \coloneqq k$.
    Then, for the constructed $X$, we have $X_{kk} = \min\{s_k, t_k\}$
    and $X_{jj} = X'_{jj} = \min\{s'_j, t'_j\} = \min\{s_j, t_j\}$
    for $j \in \numTo{n} \setminus \{k\}$.
\end{proof}

\begin{proof}[Proof of \Cref{th:colRefApproxInv}]
    Let $n \ge 2 \cdot \vs(H)$.
    Assume w.l.o.g.\ that $H$ is normalized.
    As a first step, we round the vertex weights of $H$.
    By \Cref{le:roundingVertexWeights}, we can choose $\bar{s} \in \Rnn^{V(H)}$
    such that $\sum_{u \in V(H)} s_u = n$ and
    $\lvert \frac{s_u}{n} - \alpha_u(H) \rvert < \frac{1}{n}$ for every $u \in V(H)$.
    Then, we also have
    \begin{equation*}
        \lvert \frac{n \cdot s_u}{n^2} - \alpha_u(H) \rvert < \frac{1}{n}
    \end{equation*}
    for every $u \in V(H)$.
    In the following, $n \cdot s_u$ is the size of the color class we construct
    for the vertex $u$.
    This blow-up of every color class by $n$ is crucial in the next step.
    Note that it is perfectly fine if we have $s_u = 0$ for some $u \in V(H)$
    in the following;
    we just choose the corresponding values of $M$ as $0$.

    As a second step, we round the edge weights of non-loops of $H$.
    For $u, v \in V(H)$ with $u \neq v$, we have to choose
    $M_{uv} \in \{0, \dots, n \cdot s_v\}$ and $M_{vu} \in \{0, \dots, n \cdot s_u\}$
    such that $M_{uv} \cdot n \cdot s_u = M_{vu} \cdot n \cdot s_v$.
    Note that $M_{uv} \cdot n \cdot s_u = M_{vu} \cdot n \cdot s_v$ is a common multiple of
    $n \cdot s_v$ and $n \cdot s_u$, i.e., we have
    $M_{uv} \cdot n \cdot s_u = M_{vu} \cdot n \cdot s_v = k \cdot \lcm (n \cdot s_u, n \cdot s_v)$
    for some $k \in \{0, \dots, \frac{n \cdot s_u \cdot n \cdot s_v}{\lcm(n \cdot s_u, n \cdot s_v)}\}$.
    Hence, our choice of $M_{uv}$ and $M_{vu}$ is limited to the choice of such a $k$,
    giving us
    \begin{align*}
        &M_{uv} = k \cdot \frac{\lcm(n \cdot s_u, n \cdot s_v)}{n \cdot s_u}&&\text{and}&&M_{vu} = k \cdot \frac{\lcm(n \cdot s_u, n \cdot s_v)}{n \cdot s_v}.&
    \end{align*}
    As
$\frac{\lcm(n \cdot s_u, n \cdot s_v)}{n \cdot s_u}= \frac{n \cdot \lcm(s_u, s_v)}{n \cdot s_u}= \frac{\lcm(s_u, s_v)}{s_u}\le \frac{s_u \cdot s_v}{s_u}= s_v$,
we can choose $M_{uv}$ in steps of at most $s_v$ and, symmetrically, $M_{vu}$
    in steps of $s_u$.
    This means that we can choose a $k$ such that
    \begin{equation*}
        \left\lvert \frac{M_{uv}}{n \cdot s_v} - \beta_{uv}(H) \right\rvert = \left\lvert \frac{M_{vu}}{n \cdot s_u} - \beta_{vu}(H) \right\rvert \le \frac{1}{2n}.
    \end{equation*}

    As a third step, we round the edge weights of the loops of $H$.
    For $u \in V(H)$, note that we could choose $M_{uu} \in \{0, \dots, n \cdot s_u - 1\}$ such that
    $M_{uu} \cdot n \cdot s_u$ is even and
    \begin{equation*}
        \left\lvert \frac{M_{uu}}{n \cdot s_u} - \beta_{uu}(H) \right\rvert \le \frac{1}{n \cdot s_u}.
    \end{equation*}
    However, to ensure that color refinement refines the resulting graph
    to the color classes specified by $\bar{s}$ and $M$
    and not to some coarser partition, we tweak the diagonal entries $M_{uu}$ a bit.
    The matrix $M$ is of dimension $\vs(H) \times \vs(H)$, i.e., we can obtain
    pairwise distinct row sums by choosing a value that is close to the value $M_{uu}$
    chosen above.
    More precisely, as $n \ge 2 \cdot  \vs(H)$, we always have at least $\vs(H) - 1$
    valid choices
    that deviate from the above choice of $M_{uu}$ by at most $2 (\vs(H) - 1)$.
    Hence, we can choose $M_{uu} \in \{0, \dots, n \cdot s_u - 1\}$ such that
    $M_{uu} \cdot n \cdot s_u$ is even,
    \begin{equation*}
        \left\lvert \frac{M_{uu}}{n \cdot s_u} - \beta_{uu}(H) \right\rvert \le \frac{2 (\vs(H) - 1)}{n \cdot s_u} \le \frac{2 \vs(H)}{n},
    \end{equation*}
    and all row sums of $M$ are pairwise distinct.

    By the criteria of \cite{KieferSchweitzerSelman15}, cf.\ \Cref{th:inversion},
    we obtain a graph $G$ for $(n \cdot s_u)_{u \in V(H)}$ and $M$
    on $\sum_{u \in V(H)} n \cdot s_u = n^2$ vertices
    with the corresponding partition $(C_u)_{u \in V(H)}$.
    Note that, since all row sums of $M$ are pairwise distinct,
    vertices in different sets of the partition have different degrees, i.e.,
    there is no coarser stable coloring than the one induced by
    $(C_u)_{u \in V(H)}$.
    Hence, $\ColG$ is isomorphic to the weighted graph
    $G_{n \cdot \bar{s}, M}$.

    It remains to prove that $\Col{G}$ and $H$ are actually close in the
    cut distance.
    As $(\frac{n \cdot s_u}{n^2})_{u \in V(H)} = (\frac{s_u}{n})_{u \in V(H)}$ and $(\alpha_u(H))_{u \in V(H)}$
    sum to $1$, \Cref{le:fracOverlay} yields a matrix $X \in \Rnn^{V(H) \times V(H)}$
    with
    \begin{enumerate}
        \item $\sum_{v \in V(H)} X_{uv} = \frac{s_u}{n}$ for every $u \in V(H)$,
        \item $\sum_{u \in V(H)} X_{uv} = \alpha_v(H)$ for every $v \in V(H)$, and
        \item $X_{uu} = \min \{\frac{s_u}{n}, \alpha_u(H)\}$ for every $u \in V(H)$.
    \end{enumerate}
    We have
    \begin{align*}
        \dist(\ColG, H)
        &= \dist(G_{n \cdot \bar{s}, M}, H)\nonumber\\
        &\le \fracdist(G_{n \cdot \bar{s}, M}, H, X)\nonumber\\
        &= \max_{Q, R \subseteq V(H) \times V(H)} \Big\lvert \sum_{\substack{iu \in Q,\\ jv \in R}} X_{iu} X_{jv} (\frac{M_{ij}}{n \cdot s_j} - \beta_{uv}(H)) \Big\rvert\nonumber\\
        &\le \sum_{i,j,u,v \in V(H)} X_{iu} X_{jv} \Big\lvert\frac{M_{ij}}{n \cdot s_j} - \beta_{uv}(H) \Big\rvert\nonumber\\
        &=  \begin{aligned}[t]
                \sum_{i,j \in V(H)} X_{ii} X_{jj} \Big\lvert\frac{M_{ij}}{n \cdot s_j} - \beta_{ij}(H) \Big\rvert
               + \;\;\;\sum_{\mathclap{\substack{i,j,u,v \in V(H),\\ i \neq u \text{ or } j \neq v}}}\;\; X_{iu} X_{jv} \Big\lvert\frac{M_{ij}}{n \cdot s_j} - \beta_{uv}(H) \Big\rvert.
        \end{aligned}
    \end{align*}
    For the first of these two sums, we get
    \begin{align*}
        \sum_{i,j \in V(H)} X_{ii} X_{jj} \Big\lvert\frac{M_{ij}}{n \cdot s_j} - \beta_{ij}(H) \Big\rvert &\le \sum_{i,j \in V(H)} \alpha_i(H) \alpha_j(H) \cdot \frac{2 \vs(H)}{n}\\
        &=  \frac{2 \vs(H)}{n} \cdot \sum_{i \in V(H)} \Big( \alpha_i(H) \cdot \sum_{j \in V(H)} \alpha_j(H) \Big)\\
&=  \frac{2 \vs(H)}{n}.
    \end{align*}
    For the second sum, we note that, for $u \in V(H)$, we have
    \begin{align*}
        \sum_{\substack{v \in V(H),\\ v \neq u}} X_{uv} + \sum_{\substack{v \in V(H),\\ v \neq u}} X_{vu}
        &= \sum_{v \in V(H)} X_{uv} + \sum_{v \in V(H)} X_{vu} - 2 \cdot X_{uu}\\
        &= \frac{s_u}{n} + \alpha_u(H) - 2 \cdot \min \{\frac{s_u}{n}, \alpha_u(H)\}\\
        &< \frac{1}{n}
    \end{align*}
    and, hence,
    \begin{align*}
        \sum_{\substack{u,v \in V(H),\\u \neq v}} X_{uv}= \frac{1}{2} \cdot \sum_{u \in V(H)} \Big( \sum_{\substack{v \in V(H),\\ v \neq u}} X_{uv} + \sum_{\substack{v \in V(H),\\ v \neq u}} X_{vu} \Big)
        \le \frac{1}{2} \cdot \sum_{u \in V(H)} \frac{1}{n}
        = \frac{\vs(H)}{2n}
    \end{align*}
    Then, for the second sum, we get
    \begin{align*}
        \sum_{\mathclap{\substack{i,j,u,v \in V(H),\\ i \neq u \text{ or } j \neq v}}}\;\; X_{iu} X_{jv} \Big\lvert\frac{M_{ij}}{n \cdot s_j} - \beta_{uv}(H) \Big\rvert
        &\le \;\;\;\sum_{\mathclap{\substack{i,j,u,v \in V(H),\\ i \neq u \text{ or } j \neq v}}}\;\; X_{iu} X_{jv}\\
&= \;\;\;\sum_{\mathclap{\substack{i,j,u,v \in V(H),\\ i \neq u \text{ and } j \neq v}}}\;\; X_{iu} X_{jv}
            + \;\;\;\sum_{\mathclap{\substack{i,j,v \in V(H),\\ j \neq v}}}\;\; X_{ii} X_{jv}
            + \;\;\;\sum_{\mathclap{\substack{i,j,u \in V(H),\\ i \neq u }}}\;\; X_{iu} X_{jj}\\
&= \Big(\sum_{\substack{i,u \in V(H),\\ i \neq u}} X_{iu}\Big)^2
            + 2 \cdot \Big(\sum_{i \in V(H)} X_{ii} \Big) \cdot \Big( \sum_{\substack{j,v \in V(H)\\ j \neq v}} X_{jv} \Big)\\
        &\le  \Big(\frac{\vs(H)}{2n}\Big)^2 + 2 \cdot \Big(\sum_{i \in V(H)} \alpha_i(H) \Big) \cdot \frac{\vs(H)}{2n}\\
        &= \frac{1}{4} \cdot \Big(\frac{\vs(H)}{n}\Big)^2 + \frac{\vs(H)}{n}.
    \end{align*}
    Summing up these two bounds, we get an overall upper bound of
    \begin{align*}
        3 \cdot \frac{\vs(H)}{n} + \frac{1}{4} \cdot \Big(\frac{\vs(H)}{n}\Big)^2.
    \end{align*}
\end{proof}

\end{document}